\documentclass[acmsmall]{ec22acm}

\usepackage{preample}

\AtBeginDocument{%
  \providecommand\BibTeX{{%
    \normalfont B\kern-0.5em{\scshape i\kern-0.25em b}\kern-0.8em\TeX}}}

\copyrightyear{2022} 
\acmYear{2022} 
\setcopyright{acmlicensed}\acmConference[EC '22]{Proceedings of the 23rd ACM Conference on Economics and Computation}{July 11--15, 2022}{Boulder, CO, USA}
\acmBooktitle{Proceedings of the 23rd ACM Conference on Economics and Computation (EC '22), July 11--15, 2022, Boulder, CO, USA}
\acmPrice{15.00}
\acmDOI{10.1145/3490486.3538337}
\acmISBN{978-1-4503-9150-4/22/07}

\begin{document}
\fancyhead{}
\renewcommand{\subonly}{}

\title{Optimal Strategic Mining Against Cryptographic Self-Selection in Proof-of-Stake}

\author{Matheus V. X. Ferreira}
\email{matheusventuryne@gmail.com}
\affiliation{
  \institution{Harvard University}
  \city{Cambridge}
  \state{Massachussets}
  \country{USA}
}

\author{Ye Lin Sally Hahn}
\email{yhahn@alumni.princeton.edu}
\affiliation{
  \institution{Princeton University}
  \city{Princeton}
  \state{New Jersey}
  \country{USA}
}

\author{S. Matthew Weinberg}
\email{smweinberg@princeton.edu}
\affiliation{
  \institution{Princeton University}
  \city{Princeton}
  \state{New Jersey}
  \country{USA}
}

\author{Catherine Yu}
\email{cjy@princeton.edu}
\affiliation{
  \institution{Princeton University}
  \city{Princeton}
  \state{New Jersey}
  \country{USA}
}

\begin{abstract}
Cryptographic Self-Selection is a subroutine used to select a leader for modern proof-of-stake consensus protocols. In cryptographic self-selection, each round $r$ has a seed $Q_r$. In round $r$, each account owner is asked to digitally sign $Q_r$, hash their digital signature to produce a credential, and then broadcast this credential to the entire network. A publicly-known function scores each credential in a manner so that the distribution of the lowest scoring credential is identical to the distribution of stake owned by each account. The user who broadcasts the lowest-scoring credential is the leader for round $r$, and their credential becomes the seed $Q_{r+1}$. Such protocols leave open the possibility of manipulation: a user who owns multiple accounts that each produce low-scoring credentials in round $r$ can selectively choose which ones to broadcast in order to influence the seed for round $r+1$. Indeed, the user can pre-compute their credentials for round $r+1$ for each potential seed, and broadcast only the credential (among those with low enough score to be leader) that produces the most favorable seed.

We consider an adversary who wishes to maximize the expected fraction of rounds in which an account they own is the leader. We show such an adversary always benefits from deviating from the intended protocol, regardless of the fraction of the stake controlled. We characterize the optimal strategy; first by proving the existence of optimal positive recurrent strategies whenever the adversary owns last than $\frac{3-\sqrt{5}}{2}\approx 38\%$ of the stake. Then, we provide a Markov Decision Process formulation to compute the optimal strategy.
\end{abstract}



\begin{titlepage}

\maketitle

\end{titlepage}

\section{Introduction}

Proposed by Nakamoto in 2008, Bitcoin was one of the major innovations in peer-to-peer networks for electronic transactions \cite{nakamoto2008bitcoin}. Bitcoin is a decentralized currency without an administrator where anyone is free to join and submit or validate transactions in a public ledger. To modify the state of the ledger, users must publicly broadcast transactions. Those transactions are included in a block by miners and validated via a Proof-of-Work (PoW). The significant computational resources required to validate a block coupled with the reward for validating blocks create an economic incentive for miners to validate blocks correctly.

Unfortunately, Bitcoin is not without limitations. The proof-of-work consensus mechanism was designed to be energy-intensive---the global energy consumption from all Bitcoin miners exceeds that of all but 26 countries~\cite{cbeci}. Moreover, the economies of scale from designing and purchasing large quantities of specialized hardware for proof-of-work mining has demonstrated that Bitcoin is more prone to centralization than initially thought~\cite{arnosti2022bitcoin}. In attempt to address the limitations of proof-of-work, many alternative blockchain designs have been proposed~\cite{ren2016proof, karantias2020proof, kiayias2017ouroboros, chen2019algorand}.

In particular, proof-of-stake blockchains replace proof-of-work by a randomized mechanism to select a miner as the leader who proposes a new block. To avoid sybil attacks, where an adversary impersonates multiple identities, the mechanism hopes to choose a miner with probability proportional to their fraction of owned coins (commonly referred to as their stake in the system). The main challenge for such systems is to sample a miner without sacrificing decentralization and, at the same time, preserve the security and economic properties of the blockchain.

Several proposals for ``Bitcoin-like'' longest-chain proof-of-stake protocols have been made, but several drawbacks to this approach still exist. For example,~\cite{ferreira2021proof} considers a longest-chain proof-stake blockchain that requires a randomness beacon~\cite{rabin1983transaction, kelsey2019reference}, and proves that this qualitatively preserves the mining incentives of Bitcoin, but the need for a trusted external randomness beacon is prohibitive in most settings. Without a trusted external randomness beacon, proof-of-stake implementations often rely on using the own blockchain as a source of pseudorandomness. Because miners can often predict the randomness in such protocols,~\citet{brown2019formal} showed that a large class of longest-chain proof-of-stake protocols are vulnerable to profitable deviations that they term ``predictable selfish mining.'' Thus, the current state of longest-chain proof-of-stake cryptocurrencies must either: (a) propose a trusted randomness beacon, or (b) propose clever applications of cryptography to minimize the ability of strategic miners to manipulate the blockchain pseudorandomness. While active research agendas aim to address both (a) and (b), these are currently notable barriers to longest-chain proof-of-stake protocols.

One alternative design that relies on neither a trusted beacon, nor even an underlying longest-chain protocol, is cryptographic self-selection. This procedure is adopted in blockchains like Algorand~\cite{gilad2017algorand, chen2019algorand}. In cryptographic self-selection, each round $r$ has a seed $Q_{r-1}$, used to sample the leader $\ell_r$ to propose the $r$-th block in the blockchain. In the ideal case where $Q_{r-1}$ is unbiased, a clever cryptographic construction ensures that a miner owning a $\alpha \in [0, 1]$ fraction of the coins has a probability $\alpha $ of becoming the leader.

Unlike longest-chain blockchains, blockchains using cryptographic self-selection are immutable because once the leader validates $B_r$, $B_r$ cannot be modified. Nevertheless, one limitation of such protocols is that the leader for round $r$ may have some influence over the seed $Q_{r+1}$ for round $r+1$. For example,~\citet{chen2019algorand} note that it is possible for an adversary to bias seeds in future rounds in Algorand's cryptographic self-selection. In this work, we study quantitatively the limits of how much these deviations might benefit an economically-motivated adversary. Specifically, we assume that the adversary wishes to maximize the fraction of rounds during which they are the leader. If the adversary were honest, this would be exactly an $\alpha$ fraction. We seek to understand $f(\alpha)\geq \alpha$, the maximum fraction of rounds that an adversary with an $\alpha$ fraction of the stake can lead in expectation.\footnote{This is the same objective in prior work~\cite{eyal2014majority,sapirshtein2016optimal,brown2019formal,ferreira2021proof}. In prior work, this objective function was motivated by the block reward associated with creation of each block. Even if a protocol has no block reward there is still \emph{some} economic incentive associated with creating a block. This could be due to transaction fees~\cite{ferreira2021dynamic, roughgarden2021transaction}, or side contracts that the leader is able to execute in deciding what to include. We do not explicitly model the direct connection between being a leader in a round and the monetary reward, and treat this per-block incentive as exogenous.}

\subsection{Overview of Results and Roadmap}
Our main contributions are as follows:
\begin{itemize}
    \item First, we provide a formal stochastic process that captures the game played by strategic players who want to be leader as often as possible in a cryptographic self-selection protocol. We provide a detailed, formal description of the game in Section~\ref{sec:background}, and prove several basic facts in Section~\ref{sec:facts}. These sections provide a clean stochastic process whose analysis directly informs the rewards achievable by strategic players in blockchain based on cryptographic self-selection.
    \item In this game, it is a priori possible that an \emph{extremely} strong strategy exists that lets a player with an $\alpha<1$ fraction of the stake win unboundedly-many rounds in a row in expectation. We prove that when $\alpha < \frac{3-\sqrt{5}}{2}\approx 0.38$, this is not possible, and the optimal fraction of rounds won by a strategic player with $\alpha < \frac{3-\sqrt{5}}{2}$ is $< 1$. Note that in many protocols, including~\cite{chen2019algorand}, that use cryptographic self-selection, owning $\alpha > 1/3$ of the stake is already enough to subvert consensus, so the $\alpha < 1/3$ is the most relevant range where strategic mining is a concern. We prove this in Section~\ref{sec:recurrence}.
    \item We pose a simple strategy, the $\onelookahead$ strategy, which \emph{strictly outperforms the honest strategy for all $\alpha$}. We fully analyze the expected reward of this strategy in Section~\ref{sec:warmup}.
    \item Finally, we describe how the optimal strategy can be found by solving a series of MDPs. As the MDPs are infinite, this unfortunately does not immediately give an efficient algorithm. This appears in Section~\ref{sec:opt}.
\end{itemize}

\subsection{Related Work}
Seminal work of~\cite{eyal2014majority} established that strategic mining in Bitcoin is possible. Hundreds of followup works pushed their ideas in various directions. One notable work, which is similar in spirit to ours, is~\cite{sapirshtein2016optimal}, who nail down the optimal achievable reward for a miner who has an $\alpha$-fraction of the computational power. Follow-up works such as~\cite{brown2019formal,ferreira2021proof} consider similar questions for proof-of-stake instead of proof-of-work, but to the best of our knowledge, \emph{all prior work in this agenda considers longest-chain protocols}. In comparison to this line of work, ours is the first (to our knowledge) to consider a formal model of strategic behavior in cryptographic self-selection, which is used in protocols based on Byzantine consensus.

\citet{chen2019algorand} develop the theoretical protocol for Algorand, and propose that manipulations of the cryptographic self-selection protocol may be possible. They do not propose a concrete manipulation, but do upper bound the maximum fraction of rounds that certain kinds of adversaries can possibly be leader. In comparison to their work, our work proposes a formal model to capture the entire strategy space in cryptographic self-selection protocols, and our $\onelookahead$ strategy also provides the first concrete profitable manipulation.

\section{Background and Setup}\label{sec:background}

In this section, we provide our formal model and some preliminary observations. Our model captures the cryptographic self-selection protocol of~\cite{chen2019algorand}, but we remind the reader that our model \textit{only} concerns leader selection---this process is independent of block creation, consensus, etc.

\subsection{Blockchain Protocols with Finality}\label{sec:algorand}
Many modern proof-of-stake blockchain protocols, such as Algorand~\cite{chen2019algorand, gilad2017algorand}, differ significantly from the longest-chain protocols like Bitcoin. All blockchain protocols maintain a ledger, which is a sequence of blocks $B_0, B_1, \ldots, B_{t}, \ldots$. Protocols with \emph{finality} differ from Bitcoin in that there are \emph{no forks}. Once round $t$ has concluded, there is a single well-defined $B_t$, which will stay fixed throughout eternity. 

To produce the block for round $t$, blockchain protocols with finality run an underlying consensus protocol. These consensus protocols often require a \emph{leader} $\ell_t$, selected based on $B_1,\ldots, B_{t-1}$, who gets to propose the block which could potentially be ratified as $B_t$. Note that because there are no forks, the leader $\ell_t$ is well-defined. 

\subsection{Cryptographic Self-Selection to Determine a Leader}\label{sec:leader}
One problem that any blockchain protocol with finality must resolve is how to determine $\ell_t$ as a function of $B_1,\ldots, B_{t-1}$, and this must be done with care. For example, if there are $N$ coins indexed from $1$ to $N$, one naive proposal might simply declare the owner of coin $t \pmod{N+1}$ to be the leader at round $t$. This is vulnerable to a predictability attack: an attacker knows exactly which coins they need to own in which rounds, and can be solely responsible for block proposal for many rounds in a row. Another naive proposal might declare the owner of coin $\hash(B_{t-1}) \pmod{N+1}$ to be the leader at round $t$. This is vulnerable to a grinding attack: $\ell_{t-1}$ has many options for the contents they include in $B_{t-1}$, and can try arbitrarily many contents until $\hash(B_{t-1})$ results in a coin they own. In general, the goal of a leader-selection protocol is to pick a leader in a manner so that:

\begin{itemize}
    \item When every user honestly follows the intended protocol, the distribution of each $\ell_t$ is proportional to the stake, and i.i.d. across rounds. That is, each user with an $\alpha$ fraction of the total stake is selected as the leader in round $t$ with probability $\alpha$, independently across rounds.
    
    \item A self-interested user has little ability to predict future rounds in which they could become the leader, or  to increase the fraction of rounds in which they are the leader.
\end{itemize}

Cryptographic self-selection is a clever approach, used by Algorand~\cite{chen2019algorand,gilad2017algorand} to select a leader. Before defining the full protocol, we need two basic tools.

\subsubsection{Tools for Cryptographic Self-Selection}
One useful cryptographic tool towards cryptographic self-selection is a {\em verifiable random function}~\cite{micali1999verifiable}. For the purposes of this paper, we'll use an \emph{ideal} verifiable random function.

\begin{definition}[Ideal Verifiable Random Function] An Ideal Verifiable Random Function (Ideal VRF) satisfies the following properties:
\begin{itemize}

\item \textbf{Setup.} There is an efficient randomized process to produce a secret key $\sk$ and a public key $\pk$ that parameterize the function $f_{\sk}(\cdot)$.

      \item \textbf{Private Computability.} There is an efficient algorithm $\mathcal{A}$ such that for all $\sk$, $\mathcal{A}(x,\sk) = f_{\sk}(x)$ (that is, $f$ can be efficiently computed with knowledge of the secret key $\sk$). 
    
     \item \textbf{Perfect Randomness.} For all inputs $x \neq y$, the random variables $f_{\sk}(x)$ and $f_{\sk}(y)$ are independent, and uniformly drawn from $[0,1]$, conditioned on knowledge of $\pk$. In particular, this implies that $f_\sk(x)$ is distributed uniformly on $[0,1]$ to any user who sees only $\pk$, even after that user has seen any number of pairs $(x_1,f_{\sk}(x_1)),\ldots, (x_k,f_{\sk}(x_k))$, where $x_i \neq x, \ \forall i$.
    
    \item \textbf{Verifiable.} There exists an efficient algorithm $V$ that takes as input $x, y, \pk$ and outputs \textsc{yes} if and only if $y = f_{\sk}(x)$. 
\end{itemize}
\end{definition}

Intuitively, an Ideal VRF allows a holder of $\sk$ to draw a random number uniformly from $[0,1]$ in a way that is unpredictable to anyone without knowledge of $\sk$, yet also in a verifiable manner. The distinction between an Ideal VRF and VRF lies in perfect randomness: it is generally not possible to have the random variables $\{f(x_1),\ldots, f(x_k)\}$ be \emph{statistically} indistinguishable from independent uniformly random draws from $[0,1]$ conditioned on $\pk$. VRFs used in practice instead provide that the distribution of $\{f(x_1),\ldots, f(x_k)\}$ be \emph{computationally} indistinguishable from independent uniformly random draws from $[0,1]$, conditioned on $\pk$.\footnote{Also, any (pseudo) random number generator used in practice produces output that is uniformly random over $\{0,1\}^\lambda$ for large $\lambda$, rather than over $[0,1]$. For simplicity of exposition, we think of $\lambda \rightarrow \infty$. This again does not affect our results, except for error that is exponentially small in $\lambda$ (due to the tiny possibility of ties).}

We omit a formal definition of (non-Ideal) VRFs, which is cumbersome and not relevant to our results. In particular, all proposed deviations work \emph{even when the protocol has access to an Ideal VRF}, and therefore they also work when the protocol instead uses a VRF. 

To have a simple example of a (non-Ideal) VRF in mind, consider any digital signature scheme and hash function. On input $x$, first, digitally sign $x$ to obtain $\sig(x)$, and then hash it (this is the VRF used in~\cite{chen2019algorand}). Indeed, with the secret key, a user can efficiently compute a digital signature of any $x$ and hash it. Similarly, correct computation of the hash function can be efficiently verified by anyone, and correct computation of the digital signature can be efficiently verified by anyone with the public key. Any input $\sig(x)$ to the hash function is mapped to a uniformly-random draw from $[0,1]$, independently of all other inputs, and the digital signature scheme ensures that anyone without the secret key cannot guess $\sig(x)$, even with knowledge of $x$, and any number of input/output pairs to $\sig$. 

A second tool we will need is a concept that enables the leader to be selected proportional to the stake, rather than uniformly at random among all accounts.

\begin{definition}[Balanced Scoring Function] A scoring function $S(\cdot, \cdot)$ is \emph{balanced} if for all $n \in \mathbb{N}$ and all $\langle \alpha_1,\ldots, \alpha_n\rangle \in \mathbb{R}_{>0}^n$:
$$\Pr_{X_1,\ldots, X_n \leftarrow U([0,1]^n)}\left[\arg\min_{i\in [n]}\left\{S(X_i,\alpha_i)\right\} = j\right] = \frac{\alpha_j}{\sum_{\ell=1}^n \alpha_\ell}.$$
Observe that, if ties in $\arg\min$ are broken lexicographically, this implies that for all $\alpha$, the distribution of $S(X,\alpha)$ when $X$ is drawn uniformly from $[0,1]$ must have no point-masses.
\end{definition}

Intuitively, one can think of $\arg\min_i \{X_i\}$ as the winner of a random process when each $X_1,\ldots, X_n$ is drawn independently from the uniform distribution on $[0, 1]$, denoted by $U([0,1])$, and each player is equally likely to win. A balanced scoring function allows us to redistribute the probability of winning to be proportional to $\alpha_i$ instead.

\subsubsection{Cryptographic Self-Selection Protocol}
Now, we define the cryptographic self-selection protocol, the leader-selection protocol analyzed throughout our paper.

\begin{definition}[Cryptographic Self-Selection Protocol A] The Cryptographic Self-Selection Protocol A (CSSPA) is the following:
\begin{itemize}
    \item Every account $i$ sets up an Ideal VRF with secret key $\sk_i$ and public key $\pk_i$. $\alpha_i \in [0, 1]$ refers to the fraction of the total stake that account $i$ owns.
    
    \item $Q_r$ denotes the seed for round $r$. The initial seed is a uniformly random number in $[0,1]$ constructed via a coin tossing protocol~\cite{cachin2001secure}.
    
    \item In round $r$, each user $i$ computes their credential $\cred_r^i:=f_{\sk_i}(Q_r)$. Every user can either broadcast, or not broadcast. $B_r$ denotes the set of users who broadcast in round $r$.\footnote{In order to focus on the relevant aspects, we assume that any broadcast is received by all other users. This is consistent with prior work that focuses on the underlying incentives, and not distributed computing~\cite{eyal2014majority, carlsten2016instability, brown2019formal, ferreira2021proof}.}
    
    \item There is a publicly-known balanced scoring function $S$. The leader $\ell_r$ for round $r$ is
    $$\arg\min_{i\in B_r}\{S(\cred_r^i,\alpha_i)\}.$$

    \item $Q_{r+1}:=\cred_r^{\ell_r}$. That is, the seed for round $r+1$ is the credential of the leader for round $r$.
\end{itemize}

\end{definition}

We note a few quick observations about CSSPA:
\begin{itemize}
    \item Aside from network/security/cryptography attacks, which are not the focus of this paper, the action space of a single account in each round is binary: broadcast your credential, or don't. A single player may own multiple accounts. Therefore, the actions a single player may take in our game is to: a) decide how to divide their stake among multiple accounts, and b) pick which subset of credentials to broadcast. 
    
    \item We'll refer to the \emph{honest strategy} as one which announces all credentials in every round.
    
    \item Assuming all players are honest, each leader is drawn i.i.d. and proportional to $\vec{\alpha}$. This follows immediately from the definition of Ideal VRF and balanced scoring function.
    
    \item Assuming that all players are honest, the protocol is robust to Sybil attacks. That is, a player who truly controls an $\alpha_i$ fraction of the total stake can put all of their funds into a single account, or split their funds arbitrarily over any number of accounts. No matter how they divide their funds, the probability that an account owned by this player is selected as leader is exactly $\alpha_i$. 
    
    \item Much analysis of CSSPA can be done agnostically to the particular balanced scoring function. For example, Proposition~\ref{prop:scoring} establishes that our analysis holds for a wide class of ``canonical'' balanced scoring functions. In particular, our analysis chooses a particularly simple balanced scoring function for the benefit of tractability, but our analysis holds for the balanced scoring function used in~\cite{chen2019algorand} as well via Proposition~\ref{prop:scoring}.
\end{itemize}

Throughout our paper, we'll use the balanced scoring function $S(x,\alpha_i):=\frac{-\ln(x)}{\alpha_i}$. This allows us to leverage basic facts about independent draws from exponential distributions.

\begin{definition}[Exponential Distribution]
The exponential distribution with rate $\alpha$ is the distribution with Cumulative Density Function (CDF) $F_\alpha(x) := 1 - e^{-\alpha x}$, for all $x \geq 0$. We refer to $\Exp{\alpha}$ as one independent sample from the exponential distribution with rate $\alpha$. For simplicity of notation in later calculations, we will denote by $\Exp{0}$ to be a point-mass at $+\infty$.
\end{definition}

Exponential distributions have many relevant properties that we remind the reader of in Appendix~\ref{app:exp}.

\begin{lemma}\label{lem:exponential} Define $S$ so that $S(x,\alpha_i):=\frac{-\ln(x)}{\alpha_i}$. Then $S(\cdot, \cdot)$ is a balanced scoring function. Moreover, when $x$ is drawn uniformly from $[0,1]$, $S(x,\alpha_i)$ is distributed according to $\Exp{\alpha_i}$.
\end{lemma}
\begin{proof}
We first show that $S(x,\alpha_i)$ is distributed according to $\Exp{\alpha_i}$ when $x$ is drawn uniformly from $[0,1]$. This follows essentially because $\frac{-\ln(x)}{\alpha_i}$ is the inverse reverse-CDF of $\Exp{\alpha_i}$. To see the claim, we compute the probability that $S(x,\alpha_i) > y$, for any $y$:

\begin{align*}
    \Pr_{x \leftarrow U([0,1])}[S(x,\alpha_i) > y] &= \Pr_{x \leftarrow U([0,1])}[\frac{-\ln(x)}{\alpha_i}>y]\\
    &=\Pr_{x \leftarrow U([0,1])}[\ln(x)<-\alpha_i \cdot y]\\
    &=\Pr_{x \leftarrow U([0,1])}[x<e^{-\alpha_i \cdot y}]\\
    &=e^{-\alpha_i \cdot y}
\end{align*}

This means that the CDF of $S(x,\alpha_i)$, when $x$ is drawn uniformly from $[0,1]$, is exactly $1-e^{-\alpha_i x}$, and therefore this distribution is equal to $\Exp{\alpha_i}$. Now, the fact that $S(\cdot, \cdot)$ is a balanced scoring function follows from Corollary~\ref{cor:minExps} (which states that the minimum of $X_1,\ldots, X_n$, when each $X_i$ is drawn independently from $\Exp{\alpha_i}$ is equal to $X_i$ with probability $\alpha_i$, for all $i$).
\end{proof}

We conclude this section by formally establishing that our analysis extends to a broad class of scoring functions.

\begin{definition}
A scoring function $S$ is \emph{canonical} if:
\begin{itemize}
    \item For all $\alpha$, $S(\cdot,\alpha)$ is monotone decreasing on domain $(0,1)$.
    \item For all $n$, and $\alpha_1,\ldots, \alpha_n$, the random variables $S(X,\sum_{i=1}^n \alpha_i)$ and $\min_{i=1}^n \{S(X_i,\alpha_i)\}$ are identically distributed when each $X,X_1,\ldots, X_n$ are i.i.d.~from $U([0,1])$.
    \item $S$ is \emph{continuous} in $\alpha$. That is: for all $x$ and $\alpha$, if $\lim_{\beta \rightarrow \alpha} S(x,\beta)$ exists, then $S(x,\alpha) = \lim_{\beta \rightarrow \alpha} S(x,\beta)$.
\end{itemize}
\end{definition}

Before proceeding, we give quick context for each bullet. Balanced scoring functions where $S(\cdot,\alpha)$ is not monotone decreasing exist, but assuming that $S(\cdot,\alpha)$ is monotone decreasing is w.l.o.g. Indeed, for any $\alpha$, let $F_\alpha$ denote the CDF of the random variable $S(X,\alpha)$ when $X$ is drawn uniformly from $[0,1]$. Now consider redefining $S'(x,\alpha):=F_\alpha^{-1}(1-x)$. Then the distribution of $S(X,\alpha)$ and $S'(X,\alpha)$ are identical, but $S'(\cdot,\alpha)$ is monotone decreasing. We conjecture that all balanced scoring functions satisfy the second two bullets, but we suspect that rigorously establishing this will require significant analysis. As this is not the focus of our paper, we instead treat these bullets as reasonable assumptions. Indeed, the balanced scoring function used by Algorand is canonical.

\begin{proposition}\label{prop:scoring} The game induced by CSSPA with a canonical balanced scoring function is independent of the particular canonical balanced scoring function used. Formally, for two distinct canonical balanced scoring functions $S, S'$, the games induced by CSSPA are identical. Specifically, for all players $i$, there is a bijective mapping $f$ from strategies of player $i$ in the CSSPA with $S$ to strategies of player $i$ in the CSSPA with $S'$. For all $i$, the payoff that player $i$ receives in the CSSPA with $S$ under strategy profile $\vec{s}$ is equal to the payoff that $i$ receives in the CSSPA with $S'$ under strategy profile $\langle f_i(s_i)\rangle_{i}$.
\end{proposition}

A complete proof of Proposition~\ref{prop:scoring} appears in Appendix~\ref{app:background}.

\section{Our Model: Strategic Mining in Cryptographic Self-Selection}
This section formally defines our model and, in particular, the optimization problem considered by a strategic player. Like prior work~\cite{eyal2014majority, carlsten2016instability, ferreira2021proof}, we consider a \emph{single} strategic player who is best responding to a profile of honest players. The purpose of this analysis, like in prior work, is to understand the maximum disruption that can be caused when a $1-\alpha$ fraction of the stake is owned by honest players, and an $\alpha$ fraction of the stake is owned by strategic players.\footnote{That is, the worst-case scenario, among all scenarios where a $1-\alpha$ fraction of the stake is held by honest players, is when there is a single-strategic player with an $\alpha$ fraction of the stake. Our goal, like prior work, is to understand this scenario.} We now formalize the strategy space of the strategic player.

\begin{definition}[Strategy Space in CSSPA] CSSPA is parameterized by $\alpha$, the fraction of stake owned by the strategic player, $\vec{\alpha}$ the distribution of remaining stake among honest players, and $\beta \in [0,1]$, the network connectivity strength of the strategic player. We'll refer to this as a $\beta$-strong player. When $\beta=1$, we'll simply refer to the player as \emph{strong}, and when $\beta = 0$ we'll refer to the player as \emph{weak}. The strategic player knows $\alpha, \beta, \vec{\alpha}$. 

In round $r$, the strategic player in CSSPA has the following information and makes the following decisions, in order:
\begin{enumerate}
\item The strategic player can distribute their total stake of $\alpha$ arbitrarily among as many accounts as they desire. Refer to this set as $A$.
    \item The strategic player knows $Q_r$, and knows that all other players are honest.
    
    \item For a set of accounts $B$ such that $B \cap A = \emptyset$, and $\sum_{j \in B} \alpha_j = \beta\cdot (1-\alpha)$, the strategic player learns $\cred_r^j$, for all accounts $j\in B$. The strategic player \emph{does not learn} $\cred_r^j$ for any $j \notin A \cup B$ (that is, the player only knows that each $S(\cred_r^j,\alpha_j)$ will be drawn from $\Exp{\alpha_j}$, independently).
    
    \item Observe that the strategic player can compute $\cred_r^i$ and also $S(\cred_r^i,\alpha_i)$, for all accounts $i\in A$.
    
    \item Observe further that for all $j \in A\cup B$, $\cred_r^j$ is a possible seed for $Q_{r+1}$. So the player can \emph{also} pre-compute a hypothetical $\cred_{r+1}^i$, assuming $Q_{r+1} = \cred_r^j$, for each account $i\in A$ and $j \in A \cup B$. But observe that the strategic player \emph{cannot} execute this computation for $i \notin A$ (because they cannot compute the ideal VRF for accounts $\notin A$).
    
    \item More generally, for any $k$, and any list of accounts $\langle i_0,\ldots, i_k\rangle$ such that $i_0 \in A \cup B$, and each $i_j \in A$ for all $j > 0$, the player can \emph{also} pre-compute what $\cred_{r+k}^{i_k}$ would be, assuming that $\ell_{r+j} = i_j$ for all $j \in \{0,\ldots, k-1\}$.
    
    \item The strategic player selects a subset $A_r\subseteq A$, and broadcasts all credentials in $A_r$.
\end{enumerate}

\end{definition}

We will consider optimal strategies for all $\alpha, \beta, \vec{\alpha}$. Note that the role of $\beta$ differentiates how much information they know about other players' credentials before deciding which credentials of their own to broadcast. Before getting into our main analysis, we prove some basic facts about optimal strategies in this model.

\subsection{Basic Facts on Optimal Strategies}\label{sec:facts}
First, we define the reward achieved by a particular strategy $\pi$, which the strategic player aims to optimize. A priori, the reward can depend on $\alpha, \beta$, and the distribution of the remaining $(1-\alpha)$ fraction of stake, $\vec{\alpha}$.

\begin{definition}[Reward of a Strategy] A strategy $\pi$ prescribes an action to take during each round. $X^{\alpha,\beta,\vec{\alpha}}_r(\pi)$ is an indicator random variable for the event that the strategic player is the leader during round $r$, when the game with parameters $\alpha, \beta, \vec{\alpha}$ is played. The reward of a strategy $\pi$ is simply the expected fraction of rounds where the strategic player is the leader. We drop the superscript and write $X_r(\pi)$ whenever $\alpha, \beta, \vec{\alpha}$ is clear from context. Formally:

\begin{equation}\label{eq:revenue}
    \rev_{\alpha,\beta,\vec{\alpha}}(\pi) = \e{\liminf_{R \to \infty} \frac{\sum_{r = 1}^R X^{\alpha,\beta,\vec{\alpha}}_r}{R}}
\end{equation}

The expectation is taken over the randomness in the Ideal VRFs in every round, assuming that all non-strategic miners are honest. We use the notation $\val(\alpha,\beta,\vec{\alpha}):=\sup_\pi\{\rev_{\alpha,\beta,\vec{\alpha}}(\pi)\}$. We say that a strategy $\pi$ is $\varepsilon$-optimal for parameters $\alpha,\beta,\vec{\alpha}$ if $\rev(\pi) \geq \val(\alpha,\beta,\vec{\alpha})-\varepsilon$.
\end{definition}

Next, we produce a series of refinements concerning $\varepsilon$-optimal strategies, which will allow us to greatly simplify the analysis of strategies in CSSPA. First, we observe that the strategic player need not consider any set with $|A_r| > 1$. 

\begin{observation}For any strategy $\pi$, there is another strategy $\pi'$ that results in exactly the same leaders as $\pi$ in every round (and therefore has $\rev_{\alpha,\beta,\vec{\alpha}}(\pi') = \rev_{\alpha,\beta,\vec{\alpha}}(\pi)$) for all $\alpha,\beta,\vec{\alpha}$), and always selects an $A_r$ with $|A_r| \leq 1$.
\end{observation}
\begin{proof}
Observe that the strategic player can compute $S(\cred_r^i,\alpha_i)$ for all $i \in A$. If they broadcast a set $A_r \neq \emptyset$, then the leader will be $i^*:= \arg\min_{i \in A_r}\{S(\cred_r^i,\alpha_i)\}$ if and only if $S(\cred_r^{i^*},\alpha_{i^*})< S(\cred_r^j,\alpha_j)$ for all $j \notin A$. Observe that this is exactly what would happen if the strategic player instead broadcast only $\{i^*\}$ instead. So $\pi'$ will broadcast only $\{i^*\}$, and this results in the same leader as using $\pi$.

If instead the strategic player chooses to broadcast $A_r = \emptyset$, then $\pi'$ will broadcast $\emptyset$ as well. This clearly results in the same leader as using $\pi$, because the actions are identical.

The leader is the same in both cases, and $\pi'$ only ever broadcasts (at most) a single credential.
\end{proof}

Next, we show that optimal strategies split their stake among as many accounts as possible.

\begin{lemma}\label{lemma:wallet-partitioning}
Consider a strategy $\pi$ where strategic player divides their stake into $n$ wallets with stake $\alpha_i > 0$, for $i \in [n]$. Then there is a strategy $\pi'$ where the strategic player instead divides their stake into $2n$ wallets with stake $\alpha'_i > 0$, for all $i \in [2n]$, and $\rev(\pi') = \rev(\pi)$.
\end{lemma}
\begin{proof}
The strategy $\pi'$ defines $2n$ wallets with stake
\begin{align*}
    \alpha_i' = \begin{cases}
    \frac{\alpha_i}{2} & \text{for $i \leq n$,}\\
    \frac{\alpha_{i-n}}{2} & \text{for $n < i \leq 2n$.}
    \end{cases}
\end{align*}
Observe that, conditioned on $Q_r$, $S(\cred_r^i,\alpha_i)$ is distributed according to $\Exp{\alpha_i}$, independently for all $i$. Similarly, $S(\cred_r^j,\alpha'_j)$ is distributed according to an $\Exp{\alpha'_j}$, independently for all $j$. Define now the random variable $j(i):=\arg\min\{S(\cred_r^i,\alpha'_i),S(\cred_r^{n+i},\alpha'_{n+i})\}$, and denote by $Y_r^i:=S(\cred_r^{j(i)},\alpha'_{j(i)})$. Then by Lemma~\ref{lemma:min-exp}, $Y_r^i$ is distributed according to $\Exp{\alpha'_i+\alpha'_{n+i}} = \Exp{\alpha_i}$, independently for all $i \in [n]$. Therefore, $Y_r^i$ and $S(\cred_r^i,\alpha_i)$ are \emph{identically distributed}. Therefore, we can couple executions of $\pi$ and $\pi'$ so that $Y_r^i=S(\cred_r^i,\alpha_i)$ for all $r,i$, and also so that $S(\cred_r^j,\alpha_j)$ is identical for all $r,j \notin A$. 

Consider now the strategy $\pi'$ that does the following. If $\pi$ does not broadcast a credential, then $\pi'$ also does not broadcast a credential. If $\pi$ broadcasts $i^*$, then $\pi'$ broadcasts $j(i^*)$. Observe now that the score of the credential broadcast by $\pi$ and $\pi'$ is identical (due to the coupling), and the scores of credentials broadcast by the honest players are also identical. Therefore, $\ell_r = i^*$ under $\pi$ if and only if $\ell_r = j(i^*)$ under $\pi'$. Moreover, $Q_{r+1}$ is identical under both executions. We have therefore coupled the executions of $\pi$ and $\pi'$ so that $X^{\alpha,\beta,\vec{\alpha}}_r(\pi) = X^{\alpha,\beta,\vec{\alpha}}_r(\pi') $  for all $r$, and therefore $\rev^{\alpha,\beta,\vec{\alpha}}(\pi)= \rev^{\alpha,\beta,\vec{\alpha}}(\pi')$.
\end{proof}

Next, we argue that is w.l.o.g. to consider two honest players, one with a fraction $\beta\cdot (1-\alpha)$ of the stake, and the other with fraction $(1-\beta)\cdot (1-\alpha)$ of the stake.

\begin{observation} For any $\alpha,\beta$, define $\vec{\alpha}'$ to have two honest players, one with $\alpha'_1 = \beta \cdot (1-\alpha)$, and another with $\alpha'_2 =(1-\beta)\cdot (1-\alpha)$. Then for any strategy $\pi$, $\rev^{\alpha,\beta,\vec{\alpha}}(\pi) = \rev^{\alpha,\beta,\vec{\alpha}'}$.
\end{observation}
\begin{proof}
Let $Y_B:=\min_{j \in B}\{S(\cred_r^j,\alpha_j)\}$, and $Y_C:=\min_{j \notin (A \cap B)}\{S(\cred_r^j,\alpha_j)\}$. Then by Lemma~\ref{lemma:min-exp}, $Y_B$ is distributed according to $\Exp{\beta\cdot (1-\alpha)}$, and $Y_{C}$ is distributed according to $\Exp{(1-\beta)\cdot (1-\alpha)}$. Therefore, we can couple $Y_B$ and $Y_C$ in the execution with $\vec{\alpha}$ with $S(\cred_r^1,\alpha'_1)$ and $S(\cred_r^2,\alpha'_2)$ in the execution with $\vec{\alpha}'$. 

Observe, now, that the seed for round $r+1$ in the execution with $\vec{\alpha}$ will be the minimum of $Y_B, Y_C$, and the score of the credential broadcast by the strategic player. In the execution with $\vec{\alpha}'$, the seed for round $r+1$ will be the minimum of $S(\cred_r^1,\alpha'_1)$, $S(\cred_r^2,\alpha'_2)$, and the score of the credential broadcast by the strategic player. Therefore, $Q_{r+1}$ is the same in both executions. Moreover, we also have $X_r^{\alpha,\beta,\vec{\alpha}}(\pi) = X_r^{\alpha,\beta,\vec{\alpha}'}(\pi) $. We have therefore coupled the executions with $\vec{\alpha}$ and $\vec{\alpha}'$ so that $X_r^{\alpha,\beta,\vec{\alpha}}(\pi) = X_r^{\alpha,\beta,\vec{\alpha}'}(\pi) $ for all $r$, and therefore $\rev^{\alpha,\beta,\vec{\alpha}}(\pi) = \rev^{\alpha,\beta,\vec{\alpha}'}(\pi)$. 
\end{proof}

We make one final observation, which will simplify later definitions (it is not necessary for our analysis, but greatly simplifies Definition~\ref{def:stopping}). 

\begin{observation}
For any strategy $\pi$, there is another strategy $\pi'$ satisfying $\rev^{\alpha,\beta,\vec{\alpha}}(\pi) = \rev^{\alpha,\beta,\vec{\alpha}}(\pi')$ for all $\alpha,\beta,\vec{\alpha}$, and also such that in any round $r$ where the player learns in step (3) that $\min_{j \in B}\{S(\cred_r^j,\alpha_j)\} < \min_{j \in A}\{S(\cred^r_j,\alpha_j)\}$, $\pi'$ does no computation in steps (4)-(6).
\end{observation}
\begin{proof}
Observe that if $\min_{j \in B}\{\cred_r^j\} < \min_{j \in A}\{\cred_r^j\}$, then the seed $Q^{r+1}$ will be equal to the minimum credential among all honest nodes, no matter what the strategic player chooses to broadcast. So no matter what they do this round, they cannot affect $Q^{r+1}$. Because the strategic player's actions during round $r$ have no impact on the game, they can shift any computation they originally planned to do in round $r$ later to round $r+1$. This results in a strategy $\pi'$ that results in identical seeds in every round as $\pi$, but that does no computation during rounds where their action has no impact.
\end{proof}

We now state the strategy space of the refined CSSPA.

\begin{definition}[Refined CSSPA] The refined CSSPA is parameterized by $\alpha$, the fraction of stake owned by the strategic player, and $\beta \in [0,1]$, the network connectivity strength of the strategic player. There are two honest players $B$ and $C$. $B$ owns a $\beta \cdot (1-\alpha)$ fraction of the stake, and $C$ owns a $(1-\beta)\cdot (1-\alpha)$-fraction of the stake.

In round $r$, the strategic player in CSSPA has the following information and makes the following decisions, in order:
\begin{enumerate}
    \item The strategic player can distribute their total stake of $\alpha$ arbitrarily among as many accounts as they desire. Refer to this set as $A$.
    
    \item The strategic player knows $Q_r$, and knows that all other players are honest.
    
    \item The strategic player learns $\cred_r^B$. The strategic player \emph{does not learn} $\cred_r^C$ (that is, the player only knows that $S(\cred_r^C,(1-\beta)\cdot(1-\alpha))$ will be drawn from $\Exp{(1-\beta)\cdot(1-\alpha)}$, independently).
    
    \item \label{step:csspa-4} Observe that the strategic player can compute $\cred_r^i$ and also $S(\cred_r^i,\alpha_i)$, for all accounts $i\in A\cup \{B\}$. For any $k$, and any list of accounts $\langle i_0,\ldots, i_k\rangle$ such that $i_0 \in A \cup \{B\}$ and $i_j \in A$ for all $j > 0$, the player can \emph{also} pre-compute what $\cred_{r+k}^{i_k}$ would be, assuming that $\ell_{r+j} = i_j$ for all $j \in \{0,\ldots, k-1\}$. If the strategic player learned in Step (3) that $S(\cred_r^B,\beta\cdot(1-\alpha)) < \min_{j \in A} \{S(\cred_r^j,\alpha_j)\}$, then the player does no computation.
    
    \item The strategic player selects an account $i^*$ to broadcast, or decides not to broadcast. 
\end{enumerate}
We let $\rev_{\alpha,\beta}(\pi)$, $\val(\alpha,\beta)$ denote the reward of a strategy $\pi$ in the refined CSSPA, and the optimal reward, respectively. 
\end{definition}

Based on the observations in this section, we conclude the following:

\begin{corollary} For all $\alpha,\beta, \vec{\alpha},\pi$: $\rev_{\alpha,\beta}(\pi) = \rev_{\alpha,\beta,\vec{\alpha}}(\pi)$. Therefore, $\val(\alpha,\beta)=\val(\alpha,\beta,\vec{\alpha})$ as well.
\end{corollary}

\section{Existence of Optimal Recurrent Strategies}\label{sec:recurrence}

Recall we bootstrap the initial seed $Q_0$ to be drawn from $U[0, 1]$ via a distributed coin tossing protocol. Hence $Q_0$ is an unbiased seed since it does not favor any player. Formally, we say a seed $Q_{r-1}$ is {\em unbiased} if substituting $Q_{r-1}$ by a fresh independent sample from $U[0, 1]$ results in the same distribution for $X_r(\pi), X_{r+1}(\pi), \ldots$ conditioned on all the queries to $f_{\sk_i}$ for all $i$ up to round $r-1$. Another interpretation is that the adversary did not query any $f_{\sk_i}$ on $Q_{r-1}$ before round $r$ begins which suggests the adversary is indifferent about replacing $Q_{r-1}$ for a fresh sample from $U[0, 1]$.

The adversary has a probability at most $\alpha$ of becoming the leader for round $r$ if $Q_{r-1}$ is unbiased because the probability an honest miner samples the lowest scoring credential is equal to $1-\alpha$---the adversary can only reduce their chances of being a leader by not broadcasting their credentials.

{\em How can the adversary build a biased $Q_r$ provided $Q_{r-1}$ is unbiased?} For some intuition, suppose $\beta = 1$, and the adversary has the lowest scoring credential for round $r$. In other words, the adversary observes the credentials of all honest miners and knows that if they broadcast some credential $\cred_r^{i^*}$, $i^* \in A$ becomes the leader for round $r$. However, the adversary also has the option to not broadcast any credential, in which case, some account $B$ becomes the leader. Note that the adversary already knows $\cred_r^{B}$ before deciding if they will broadcast $\cred_r^{i^*}$ or not (the assumption $\beta = 1$ implies the adversary is well connected and get to see all other credentials before taking any action). Then, the adversary queries $f_{\sk_i}$ on $\cred_r^{i^*}$ and $\cred_r^{B}$ for all $i \in A$ and observes which seed would be more favorable for round $r+1$ (would allow the adversary to sample credentials with the lowest scores for round $r+1$). This concludes our example, and in Section~\ref{sec:warmup}, we provide a complete description of one such strategy. As a takeaway, the {\em the adversary can bias the seed $Q_r$ unless the credential with the lowest score comes from an account $j \notin A$}.

It will be convenient to ask when the game reaches a round $\tau \geq 1$ where $Q_{\tau+1}$ is unbiased given that $Q_0$ is unbiased. 

\begin{definition}[Stopping Time]\label{def:stopping} We call a round $\tau$ a \emph{stopping time} if for all possible strategies $\pi$, the distribution of $\{X_r(\pi)\}_{r > \tau}$, conditioned on $Q_\tau$ and all information the adversary has during round $\tau$, is identical to the distribution of $\{X_r(\pi)\}_{r > \tau}$ after replacing $Q_{\tau+1}$ with a uniformly random draw from $[0,1]$. That is, $\tau$ is a stopping time if the game effectively resets at round $\tau+1$, because the adversary was unable to bias the distribution of $Q_{\tau+1}$. 
\end{definition}

We now state the main way in which stopping times arise.

\begin{observation}
Let $\tau$ be a round such that the adversary does not query any VRF on $Q_{\tau+1}$ during any round $\leq \tau$. Then $\tau$ is a stopping time.
\end{observation}
\begin{proof}
Because the adversary has not queried $Q_{\tau+1}$ on any VRF, this means that the adversary currently believes that every future query to any VRF on $Q_{\tau+1}$ is independently drawn from $U([0,1])$ (by definition of VRF). Replacing $Q_{\tau+1}$ with any other seed that has not been queried by the adversary has exactly the same distribution. In particular, with probability $1$, a uniformly random draw from $[0,1]$ has not been queried by the adversary in any previous round, and therefore $\tau$ is a stopping time.
\end{proof}

\begin{definition}[Positive Recurrence]Let $\tau \geq 1$ be the first stopping time induced by $\pi$. We say $\pi$ is {\em positive recurrent} if $\e{\tau} < \infty$.
\end{definition}

Let $\tau_0, \tau_1, \ldots$ be a sequence of stopping times. Since we can assume the adversary's strategy resets whenever a stopping time is reached, $\tau_1 - \tau_0, \tau_2 - \tau_1, \ldots$ and $\sum_{r = \tau_0+1}^{\tau_1} X_r(\pi), \sum_{r = \tau_1 + 1}^{\tau_2} X_r(\pi), \ldots$ are sequences of i.i.d. random variables. The following result simplifies the expression for revenue for positive recurrent strategies:

\begin{lemma}\label{lemma:rev-recurrence}
Let $\pi$ be positive recurrent. Then $\rev(\pi) = \frac{\e{\sum_{r = 1}^\tau X_r(\pi)}}{\e{\tau}}$ where $\tau$ is a stopping time.
\end{lemma}
\begin{proof}
Let $\tau_0 = 0, \tau_1, \tau_2, \ldots$ be the sequence of stopping times and let $N(t)$ be the index for the most recent stopping time by time $t$. Then
\begin{align*}
    \rev(\pi) = \e{\liminf_{T \to \infty}\frac{\frac{1}{N(T)}\left(\sum_{r = \tau_{N(T)}}^T X_r(\pi) + \sum_{i = 1}^{N(T)} \sum_{r = \tau_{i-1} + 1}^{\tau_i} X_r(\pi)\right)}{\frac{1}{N(T)}\left((T - \tau_{N(T)}) + \sum_{i = 1}^{N(T)} (\tau_i - \tau_{i-1})\right)}}
\end{align*}
Since $N(T) \to \infty$ as $T \to \infty$, the statement follows from the strong law of large numbers (Lemma~\ref{lemma:slln}).
\end{proof}

Lemma~\ref{lemma:rev-recurrence} provides a nice characterization for the revenue of positive recurrent strategies which will be critical when studying optimal strategies. In the rest of this section, we aim to show a sufficient condition for the existence of optimal positive recurrent strategies by proving the following informal claim: {\em for any strategy $\pi$, let $\tau \geq 1$ be the first round where $\arg\min_{i \in [n]} \score(\cred_\tau^i, \alpha_i) \notin A$, then $\tau$ is a stopping time and $\e{\tau} < \infty$ for $\alpha < \frac{3-\sqrt{5}}{2} \approx 0.38$}.

\begin{definition}[Forced stopping time]
Consider round $r$ with seed $Q_{r}$. If $\arg\min_{i \in [n]} \score(\cred_r^i, \alpha_i) \not\in A$, we say $r$ is a {\em forced stopping time} with respect to $Q_{r}$.
\end{definition}

\begin{lemma}\label{lemma:stopping-time}
If $r$ is a forced stopping time, then $r$ is a stopping time.
\end{lemma}
\begin{proof}
The leader $\ell_r \notin A$, because both $B$ and $C$ always broadcast their credentials, and one of them has the lowest score. Let $j^*$ refer to the account in $\{B,C\}$ with minimum score. Then $Q_r = \cred_r^{j^*}$ regardless of the adversary's action.

Now, observe that the probability that $Q_{r+1}$ has been any previous credential in any previous round $< r$ is $0$ (because all credentials are drawn uniformly from $[0,1]$ when drawn). Moreover, \emph{because $\ell_r \notin A$}, the adversary cannot possibly have known $Q_{r+1}$ prior to round $r$. This is because the adversary cannot compute the VRF of $\ell_r$, and $\ell_r$ only broadcasts $Q_{r+1}$ during round $r$. Finally, the adversary did not query $Q_{r+1}$ after learning $Q_{r+1}$ during round $r$ because either the minimum account was $B$ (in which case, by definition of step (\ref{step:csspa-4}), the adversary did not query $B$), or the minimum account was $C$ (in which case, the adversary does not have a step to query any VRFs during round $r$ after learning $C$. Therefore, the adversary certainly did not query $Q^{r+1}$ after learning $Q^{r+1}$. 

The only remaining possibility is that the adversary had previously decided to query $Q^{r+1}$ at a point when all they know is that $Q^{r+1}$ is drawn independently from $U([0,1])$, conditioned on inducing the minimum credential for round $r$. As this distribution is continuous (even after any conditioning), the probability that it outputs any particular credential is $0$. Therefore, assuming that the adversary queries a finite number of inputs across all previous rounds, the probability that it has previously queried any VRF on $Q^{r+1}$ during any previous round is also $0$. 

Therefore, $Q^{r+1}$ has not been queried by the adversary in rounds $\leq r$, and $r$ is a stopping time.


\end{proof}

\subsection{The Branching Process}
Next, we aim to show that the expected value of the forced stopping time is finite whenever the adversary owns at most $38\%$ of the stake. Fix the seed $Q_{r-1}$ and let $j^* = \arg\min_{j \notin A} \score(\cred_r^j, \alpha_j)$, the honest account with lowest score when the seed is $Q_{r-1}$. Let $W(Q_{r-1})$ denote all the accounts that could become leaders during round $r$ when the seed is $Q_{r-1}$:
$$W(Q_{r-1}) = \{j^*\} \cup \{i \in A : \score(\cred_r^i, \alpha_i) < \score(\cred_r^{j^*}, \alpha_{j^*})\}$$
The distribution of $|W(Q_{r-1})|$ is related to the growth distribution in a Galton Watson branching process~\cite{watson1875probability}. To see this, consider a tree $\tree(Q_0)$ where each node stores a seed. We give a recursive definition for $\tree(Q_0)$. Initialize the tree to contain only the root $Q_0$, which we color black. Then while $\tree(Q_0)$ contains some black node $Q$:

\begin{itemize}
    \item If $|W(Q)| \geq 2$, for each $i \in W(Q)$, we append the edge $(Q, f_{\sk_i}(Q))$ to $\tree(Q_0)$. Color $Q$ red and color $f_{\sk_i}(Q)$ black.
    
    \item If $|W(Q)| = 1$, color $Q$ red.
\end{itemize}

Intuitively, a node is colored red without appending new edges whenever that node is a forced stopping time. A node is colored red after appending a new edge if it is not a forced stopping time (and then we need to recurse on each possible subgame induced by each possible seed).

The height of $\tree(Q_0)$ gives an upper bound for how long it takes for a game starting with seed $Q_{r-1}$ to reach a forced stopping time. To see this, consider an \emph{omniscient} adversary, who knows all secret keys (and therefore can query all VRFs in any round). Even this omniscient adversary can bias the next $k \geq 1$ rounds, if and only if $|W(Q_{r-1})| \geq 2$ (the adversary has at least two options for the seed $Q_r$) and there is a value for $Q_r \in \{\cred_r^i : i \in W(Q_{r-1})\}$ such that the omniscient adversary can bias the next $k-1$ rounds. In other words, the omniscient adversary can bias $k$ rounds if and only if there is a path $Q_0, Q_1, \ldots, Q_k$ in the tree. 

A real adversary cannot search over the entire tree for the longest path, since the real adversary cannot compute the VDFs of accounts they do not own in future rounds (they can still compute VDFs for their own accounts in hypothetical future rounds, which provides statistical information about what the tree might look like in future rounds, but they do not know the precise tree as the omniscient adversary does). However, the performance of the omniscient adversary is clearly an upper bound on the performance of the real adversary, so $T$ provides an upper bound for the number of rounds the adversary can bias. Hence, showing that the expected height of $T$ is finite implies that any strategy played by a strategic miner is positive recurrent.

First, we will characterize the distribution of $|W(Q_{r-1})|$. Formally, we will show that
$$\pr{|W(Q_{r-1})| - 1 = j} = \alpha^j(1-\alpha).$$
The notation $\min^{(i)}\{S\}$ refers to the $i^{th}$-smallest element of $S$ for $i \geq 1$ and $\min^{(0)} \{S\} := 0$. As a technical tool, we recall a useful property for exponential distributions: for all $k \geq 1$, $\min_{i \in [n]}^{(k)} \score(\cred_r^i, \alpha_i)$ is identically distributed to $\Exp{\alpha} + \min_{i \in [n]}^{(k-1)} \score(\cred_r^j, \alpha_k)$ where $\Exp{\alpha}$ refers to an independent sample from the exponentially distributed with rate $\alpha$. We defer the proof to Appendix~\ref{app:recurrence}.
\begin{lemma}\label{lemma:sampling}
Let $X_1, X_2, \ldots$ be i.i.d. copies of an exponentially distributed random variable such that $\min_{n \in \mathbb N} X_n$ is exponentially distributed with rate $\alpha$. Then, for all $i \in \mathbb N$, the random variable $Y_i = \min_{n \in \mathbb N}^{(i)} X_n$ is identically distributed to $Z_i = Z_{i-1} + \Exp{\alpha}$ where $Z_0 := 0$.
\end{lemma}

\noindent\textbf{Remark.} Lemma~\ref{lemma:sampling} provides an useful tool to reduce the computational cost of sampling only the best credentials for our adversary. If one wants to sample the $k$ lowest scores among accounts in $A$, a naive approach would require us to take $|A|$ samples from $\Exp{\frac{\alpha}{|A|}}$, sort in increasing order and output the first $k$ credentials. However, from Lemma~\ref{lemma:sampling}, it suffices to sample and output the sequence $X_1 = \Exp{\alpha}, X_2 = X_1 + \Exp{\alpha}, \ldots X_k = X_{k-1} + \Exp{\alpha}$.

We now prove the probability that the adversary has $j$ options for the seed of round $r$ given $Q_{r-1}$ is $\alpha^j (1-\alpha)$.
\begin{lemma}\label{lemma:winners}
Let $X_1, X_2, \ldots$ be i.i.d. exponentially distributed random variables such that $\min_{n \in \mathbb N} X_n$ is exponentially distributed with rate $\alpha$. Let $W$ be exponentially distributed with rate $1-\alpha$. Let $S = \{i \in \mathbb N : X_i < W\}$. Then $\pr{|S| = j} = \alpha^j (1-\alpha)$.
\end{lemma}
\begin{proof}
Let $Z_i = \min_{n \in \mathbb N}^{(i)} X_n$. Let $E_i$ denote the event $Z_i < W$ and let $E_i^c$ be its complement. Then $|S| = j$ if and only if $Z_1 < Z_2 < \ldots < Z_j < W < Z_{j+1}$. Then
\begin{align*}
    \pr{|S| = j} &= \pr{E_{j+1}^c \cap (\cap_{i = 1}^j E_i)}\\
    &= \pr{E_{j+1}^c | E_j} \prod_{i = 1}^j \pr{E_i | E_{i-1}} \\
    &= \pr{W < Z_{j+1} | W > Z_j} \prod_{i = 1}^j \pr{W > Z_i| W > Z_{i-1}}.
\end{align*}
From Lemma~\ref{lemma:sampling}, $Z_{i+1}$ is identically distributed to $Z_i + \Exp{\alpha}$ for all $i \in \mathbb N$. Then
\begin{align*}
    \pr{|S| = j} &= \pr{W < Z_j + \Exp{\alpha} | W > Z_j} \\
    &\qquad \times \prod_{i = 1}^j \pr{W > Z_{i-1} + \Exp{\alpha} | W > Z_{i-1}}\\
    &= \pr{W < \Exp{\alpha}} \prod_{i = 1}^j \pr{W > \Exp{\alpha}} \qquad & \text{From Lemma~\ref{lemma:memoryless},}\\
    &= (1-\alpha)\prod_{i = 1}^j \alpha = \alpha^j (1-\alpha) \quad & \text{From Lemma~\ref{lemma:exp-inequality}}
\end{align*}
\end{proof}

\begin{corollary}\label{cor:winners}
Let $Q_{r-1}$ de drawn from $U[0, 1]$ and $W(Q_{r-1}) = \{i \in A : \score(\cred_r^i, \alpha_i) < \min_{j \notin A} \score(\cred_r^j, \alpha_j)\}$. Then $\pr{|W(Q_{r-1})| = j} = \alpha^j (1-\alpha)$.
\end{corollary}
\begin{proof}
Recall $\min_{j \notin A} \score(\cred_r^j, \alpha_j)$ is identically distributed to $\Exp{1-\alpha}$ and $\min_{i \in A} \score(\cred_r^i, \alpha_i)$ is identically distributed to $\Exp{\alpha}$ (Lemma~\ref{lemma:min-exp}). From Lemma~\ref{lemma:winners}, $\pr{|W(Q_{r-1})| = j} = \alpha^j (1-\alpha)$ as desired.
\end{proof}

\subsection{Extinction in the Branching Process}
Next, we derive necessary conditions for the expected height of $\tree(Q_0)$ to be finite. This result will will imply the existence of optimal positive recurrent strategies.

\begin{lemma}\label{lemma:height-probability}
Let $Q_0$ be an unbiased seed and let $\tau$ be the first forced stopping time. Then $\pr{\tau \geq k} \leq \left(\frac{\alpha(2-\alpha)}{1-\alpha}\right)^k$.
\end{lemma}
\begin{proof}
Clearly $\tau$ is upper bounded by the height of $\tree(Q_0)$, then the event $\tau \geq k$ implies the height of $\tree(Q_0)$ is at least $k+1$. For all $k \geq 0$ and $Q \in [0, 1]$, let $E_{k, Q}$ denote the event that $\tree(Q)$ has height at least $k+1$. Note $\pr{E_{0, Q}} = 1$. Then, for $k \geq 1$, the event $E_{k, Q}$ holds if and only if $|W(Q)| \geq 2$ and for some child $Q' \in W(Q)$, the sub-tree $\tree(Q')$ has height at least $k-1$. Let $A_k = \pr{E_{k, Q_0}}$. Then,

\begin{align*}
    \pr{\tau \geq k} &\leq \pr{E_{k, Q_0}} = A_k\\
    &= \sum_{i = 1}^\infty \pr{|W(Q_0)| = i+1} \pr{\cup_{j \in W(Q_0)\cup \{j_0\}} E_{k-1, \cred_0^j}}\\
    &= (1-\alpha)\sum_{i = 1}^\infty \alpha^i \pr{\cup_{j \in W(Q_0) \cup \{j_0\} } E_{k-1, \cred_0^j} | \abs{W(Q_0)} = i+1} \quad & \text{From Lemma~\ref{lemma:winners},}\\
    &\leq (1-\alpha)\sum_{i = 1}^\infty (i+1)\alpha^i \pr{E_{k-1, U[0, 1]}}\quad & \text{From the union bound,}\\
    &= \frac{\alpha(2-\alpha)}{1-\alpha} A_{k-1}
\end{align*}
The last line observes the geometric series converges to $\frac{\alpha(2-\alpha)}{1-\alpha}$. To conclude, we proof by induction that $A_k \leq \left(\frac{\alpha(2-\alpha)}{1-\alpha}\right)^k$. The base case is clear: $A_0 \leq 1$. For $k \geq 1$, the inductive assumption gives
$$A_k = \frac{\alpha(2-\alpha)}{1-\alpha} A_{k-1} \leq \left(\frac{\alpha(2-\alpha)}{1-\alpha}\right)^k$$
as desired. This proves the statement.
\end{proof}

\begin{theorem}\label{thm:stopping-time}
Consider any strategy $\pi$ and let $Q_0$ be an unbiased seed. Let $\tau \geq 1$ the first forced stopping time. If $\alpha < \frac{3-\sqrt{5}}{2} \approx 0.38$, $\e{\tau} < \frac{1-\alpha}{1-3\alpha+\alpha^2} < \infty$. Hence $\pi$ is positive recurrent.
\end{theorem}
\begin{proof}
Recall that for positive discrete random $\tau$, $\e{\tau} = \sum_{i = 0}^\infty \pr{\tau > i}$. From Lemma~\ref{lemma:height-probability},
\begin{align*}
    \e{\tau} < \sum_{i = 0}^\infty \left(\frac{\alpha(2-\alpha)}{1-\alpha}\right)^i = \frac{1-\alpha}{1-3\alpha+\alpha^2} < \infty
\end{align*}
The last inequality observes the geometric series converges for $\alpha < \frac{3-\sqrt{5}}{2}$.
\end{proof}

As an application of Theorem~\ref{thm:stopping-time}, we derive a theoretical upper bound on the revenue for any strategy. Figure~\ref{fig:bound} compares the curve for the theoretical upper bound with the revenue of the honest strategy.

\begin{theorem}\label{thm:bound}
For $\alpha < \frac{3-\sqrt{5}}{2} \approx 0.38$, and all $\beta$, $\val(\alpha,\beta) \leq \frac{\alpha(2-\alpha)}{1-\alpha}$.
\end{theorem}
\begin{proof}
From Theorem~\ref{thm:stopping-time}, for $\alpha < \frac{3-\sqrt{5}}{2}$, there is an optimal positive recurrent strategy $\pi$. Let $\tau \geq 1$ be a forced stopping time. From Lemma~\ref{lemma:stopping-time}, $\tau$ is a stopping time and
\begin{align*}
    \rev(\pi) &= \frac{\e{\sum_{r = 1}^\tau X_r(\pi)}}{\e{\tau}} \quad & \text{Lemma~\ref{lemma:rev-recurrence},}\\
    &\leq \frac{\e{\tau - 1}}{\e{\tau}}\\
    &= 1 - \frac{1}{\e{\tau}} \quad & \text{From linearity of expectation,}\\
    &\leq \frac{\alpha(2 - \alpha)}{1-\alpha} \quad & \text{From Theorem~\ref{thm:stopping-time}}.
\end{align*}
The first inequality observes that if the adversary cannot choose $Q_\tau$, then the adversary does not create block $B_\tau$.
\end{proof}

\begin{figure}
\centering
\begin{minipage}{0.475\textwidth}
\includegraphics[width=\linewidth]{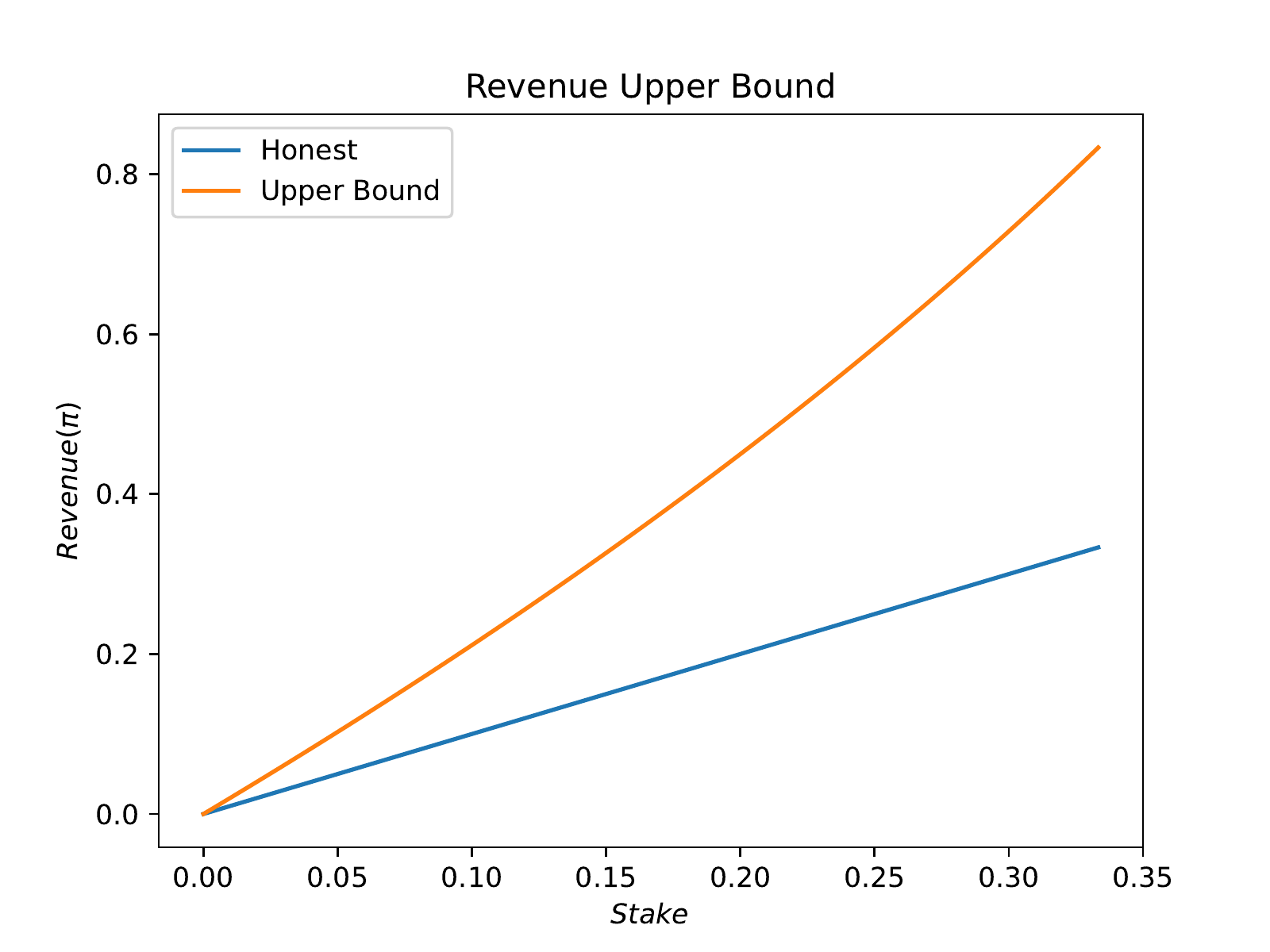}
\end{minipage}
\begin{minipage}{0.475\textwidth}
\includegraphics[width=\linewidth]{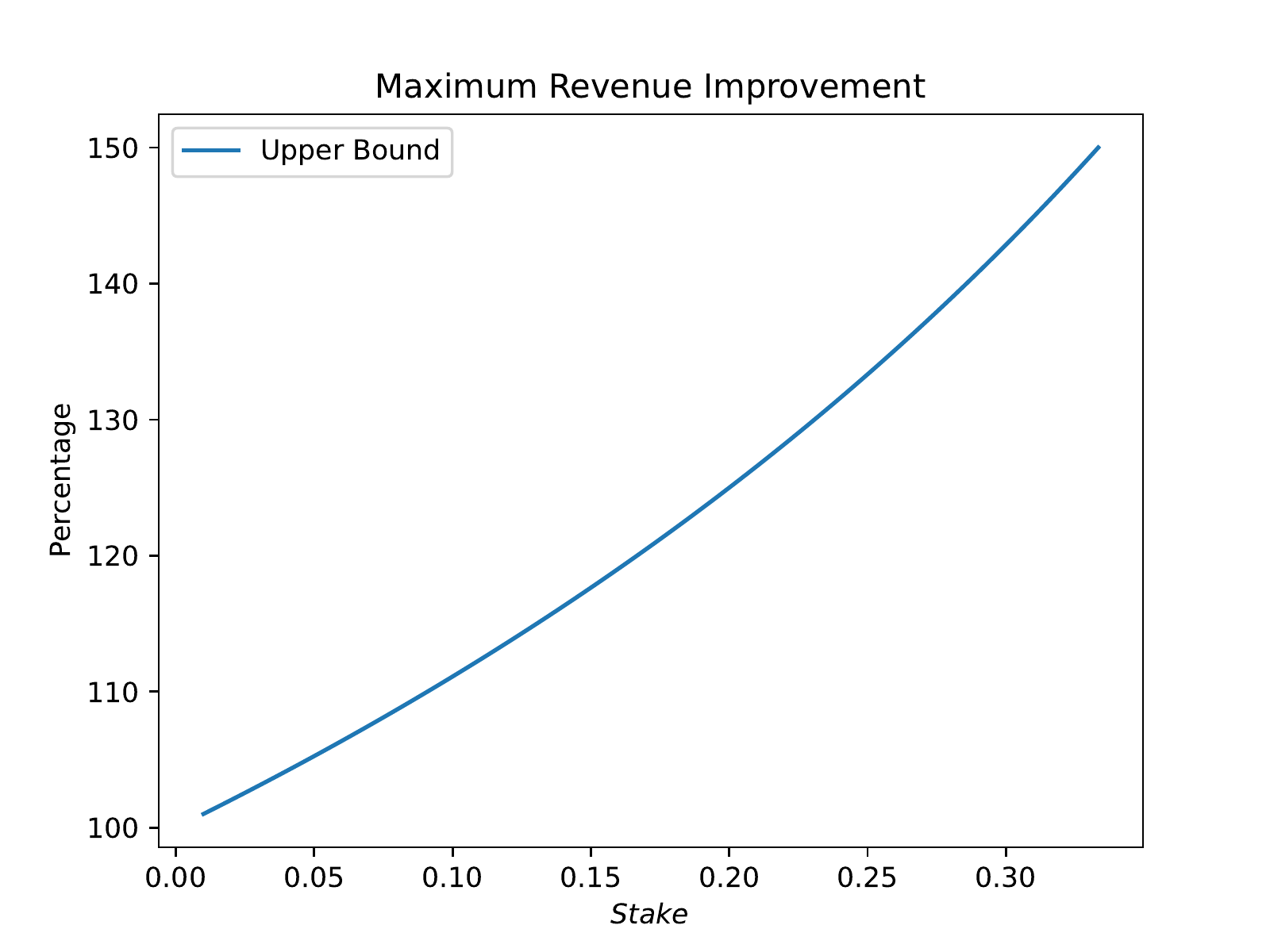}
\end{minipage}
\caption{Maximum revenue attained by any strategy. In the left, we plot the revenue for the honest strategy and our upper bound for the maximum revenue. In the right, we plot the maximum revenue improvement relative to the honest strategy.}
\label{fig:bound}
\end{figure}

The analysis in this section shows the following: for all $\alpha < \frac{3-\sqrt{5}}{2}$, even an omniscient adversary with an $\alpha$ fraction of the total stake can win at most an $\alpha \cdot \frac{2-\alpha}{1-\alpha}<1$ fraction of the rounds in expectation.

\section{The 1-Lookahead strategy}\label{sec:warmup}

This section defines the $\onelookahead$ strategy for a strong adversary ($\beta = 1$), which outperforms the honest strategy for any value of $\alpha$. Recall that the adversary divides their stake equally among an arbitrarily large number of accounts $A$. Note that this is a concrete strategy that can be used in CSSPA, and therefore its reward gives a lower bound on $\val(\alpha,1)$. 
\begin{definition}[$\onelookahead$ strategy] The strategy proceed as follows:
\begin{enumerate}
    \item \label{step:depth-1-1} Let $r$ be the current round. Let $W(Q_{r-1}) = \{i \in A: \score(\cred_r^i, \alpha_i) < \min_{i \notin A} \score(\cred_r^i, \alpha_i)\}$ be the collection of potential winners for the adversary.
    
    \item \label{step:depth-1-2} If $|W(Q_{r-1})| = 0$, broadcast no credentials. Terminate round $r$ and return to Step~\ref{step:depth-1-1}.
    
    \item \label{step:depth-1-3} If $|W(Q_{r-1})| \geq 1$, for each potential winner $i \in W(Q_{r-1})$, for each account $j \in A$, sample credential $\cred_{r+1}^{i, j} = f_{\sk_j}(\cred_r^i)$. Let $j(i) = \min_{j \in A} \cred_{r+1}^{i, j}$.
    
    \item Let $i^* = \arg \min_{i \in W(Q_{r-1})} \min_{j \in A} \score(\cred_{r+1}^{i, j}, \alpha_j)$.
    
    \item Broadcast $\cred_r^{i^*}$ at round $r$ and $\cred_{r+1}^{j(i^*)}$ at round $r + 1$.
    
    \item Return to Step~\ref{step:depth-1-1}.
\end{enumerate}
\end{definition}

\begin{theorem}\label{thm:rev-lookahead}
$\rev_{\alpha,1}(\onelookahead) = \frac{1-\alpha}{1+\alpha}\sum_{i = 1}^\infty \alpha^i\left(1 + \frac{i\alpha}{1+(i-1)\alpha}\right)$.
\end{theorem}
\begin{proof}

For our strategy, consider the stopping time $\tau$ where $\tau = 1$ when $|W(Q_0)| = 0$ and $\tau = 2$ when $|W(Q_0)| \geq 1$. Observe that this is indeed a stopping time as: (a) when $|W(Q_0)| = 0$, $\tau=1$ is a forced stopping time, and (b) while the adversary queries $f_{\sk_j}(Q_1)$ for multiple possible values of $Q_1$, \emph{they do not query $f_{\sk_j}(Q_2)$ for any of the possible values for $Q_2$}. Therefore, from the perspective of the adversary, the distribution of any VRF $Q_2$ is just $U([0,1])$, and this distribution is identical if we replace $Q_2$ with a fresh seed. Therefore, $\tau=2$ is indeed a stopping time when $W(Q_0) \geq 1$.

Now let's compute $E[\tau]$. The probability $|W(Q_0)| \geq 1$ is equal to the probability $\min_{i \in A}\{\score(\cred_r^i, \alpha_i)\} < \min_{j \notin A}\{\score(\cred_r^j, \alpha_j)\}$. The first term is exponentially distributed with rate $\alpha$ (Lemma~\ref{lemma:min-exp}) while the second is exponentially distributed with rate $1-\alpha$. Hence the probability the adversary has at least one winner is $\alpha$ (Lemma~\ref{lemma:exp-inequality}). Then $\e{\tau} = (1-\alpha) + 2\alpha  = 1 + \alpha$. Because the $\onelookahead$ strategy is positive recurrent Lemma~\ref{lemma:rev-recurrence} implies
$$\rev(\onelookahead) = \frac{\e{\sum_{r = 1}^\tau X_r(\pi)}}{\e{\tau} = 1 + \alpha}.$$
Let's compute the numerator. If $|W(Q_0)| \geq 1$, the adversary wins round $r$ since they always reveal a winning credential. Moreover, for round $r+1$, they reveal a credential with score $\min_{i \in W(Q_0), k \in A} \cred_{r+1}^{i, k}$ which is exponentially distributed with rate $\alpha \cdot |W(Q_0)|$ (Lemma~\ref{lemma:min-exp}). From Lemma~\ref{lemma:exp-inequality}, the probability the adversary wins round $r+1$ given $|W(Q_0)| = j$ is $\frac{\alpha j}{1 + \alpha (j - 1)}$. Hence,
$$\e{\sum_{r = 1}^\tau X_r(\pi) | |W(Q_0)| = j} = \ind{j \geq 1} + \frac{\alpha j}{1 + \alpha (j - 1)}.$$
From Corollary~\ref{cor:winners} the probability $|W(Q_0)| = j$ is $\alpha^j(1-\alpha)$. Then
\begin{align*}
\e{\sum_{r = 1}^\tau X_r(\pi)} &= \sum_{j = 1}^\infty \pr{W(Q_0) = j}\e{\sum_{r = 1}^\tau X_r(\pi) | W(Q_0) = j}\\
&= \sum_{j = 1}^\infty \alpha^j (1-\alpha) \left(1 + \frac{\alpha j}{1+ \alpha(j-1)}\right)
\end{align*}
as desired. This concludes the proof.
\end{proof}
Figure~\ref{fig:lookahead} shows the revenue of the $\onelookahead$ strategy (Theorem~\ref{thm:rev-lookahead}) against the revenue of the honest strategy. Observe that it is always more profitable than the honest strategy, as expected.
\begin{figure}
\centering
\begin{minipage}{0.475\textwidth}
\includegraphics[width=\linewidth]{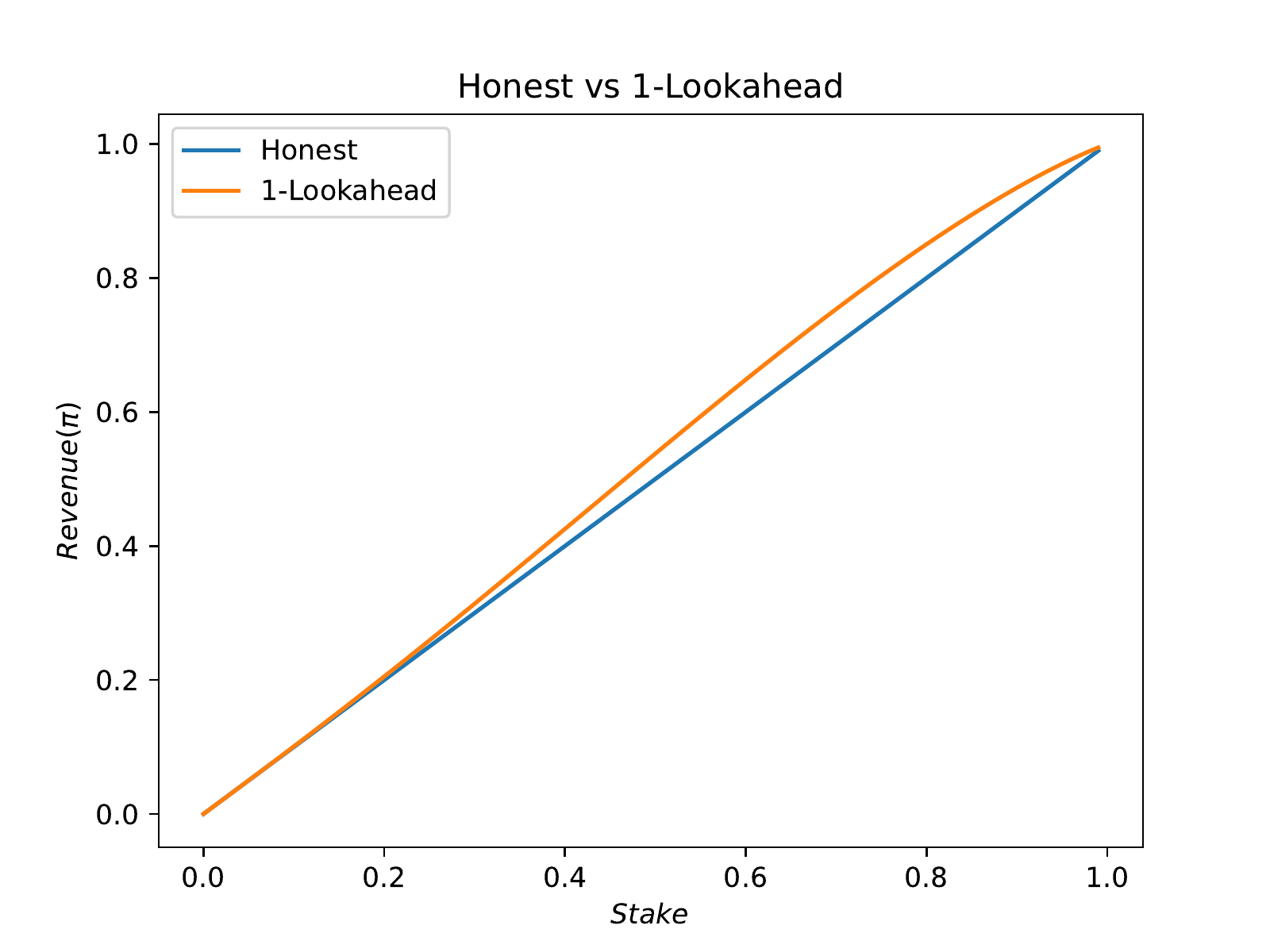}
\end{minipage}
\begin{minipage}{0.475\textwidth}
\includegraphics[width=\linewidth]{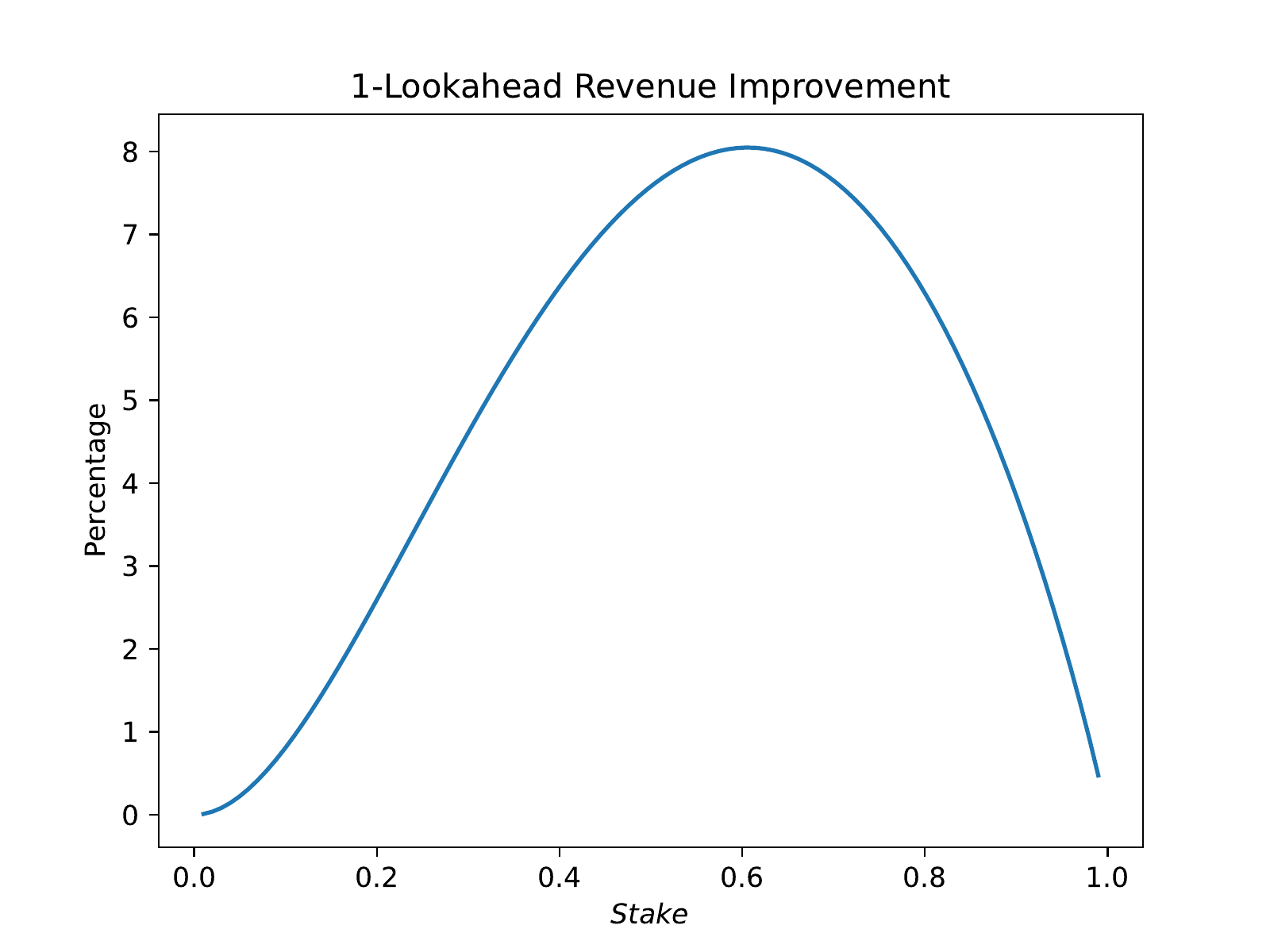}
\end{minipage}
\caption{Revenue for the $\onelookahead$ strategy. In the left, we have the absolute revenue of the honest and the $\onelookahead$ strategies. In the right, we plot the percentage revenue improvement from $\onelookahead$ relative to the honest strategy.}
\label{fig:lookahead}
\end{figure}

\section{Markov Decision Process for Optimal Strategies}\label{sec:opt}

This section shows how the optimal strategy can be computed by querying a Markov Decision Process (MDP) solver whenever $\alpha < \frac{3-\sqrt{5}}{2}\approx 0.38$. Let us recall the available information for the adversary before choosing an action. Once the game starts, the miner can compute all possible values for $Q_1, Q_2, \ldots, Q_k$ assuming $\ell_1 \in A \cup \{B\}$, $\ell_2, \ell_3 \ldots, \ell_{k-1} \in A$. As a special case, the $\onelookahead$ strategy only computes the possible values of $Q_2$ before choosing which credential to broadcast in the first round.

The information available for the adversary can be encoded in a tree where each node is a seed. For a seed $Q$, constant $k \geq 0$, define the $Q$ rooted tree $\tree_k(Q)$ recursively as follows:

\begin{itemize}
    \item If $k = 0$, let $\tree_k(Q)$ contain only the root $Q$.
    
    \item If $k \geq 1$, let $Q$ be the root of $\tree_k(Q)$. Moreover,
    
    \begin{itemize}
        \item For each $i \in A$, add edge $(Q, \tree_{k-1}(f_{\sk_i}(Q)))$ to $\tree_k(Q)$ where $\tree_{k-1}(f_{\sk_i}(Q))$ becomes a $f_{\sk_i}(Q)$ rooted sub-tree in $\tree_k(Q)$ connected by the edge $(Q, f_{\sk_i}(Q))$.
        
        \item For each $i \notin A$, once user $i$ already broadcast $f_{\sk_i}(Q)$, add edge $(Q, \tree_{k-1}(f_{\sk_i}(Q)))$ to $\tree_k(Q)$.
    \end{itemize}
\end{itemize}

Let $\tree(Q)$ be the graph obtained when we take $k \to \infty$ in $\tree_k(Q)$. Recall the basic facts for an optimal strategy $\pi$ from Section~\ref{sec:facts}: (1) $\pi$ divides its stake $\alpha$ among an infinite amount of wallets; (2) $\pi$ broadcast at most one credential each round. Then, without loss of generality, a strategy $\pi$ maps $\tree(Q)$ to at most one credential from $\{f_{\sk_i}(Q)\}_{i \in A}$, corresponding to the credential $\pi$ broadcast in a round with seed $Q$. If the strategy outputs no credential, we write $\pi(\tree(Q)) = \bot$.

\begin{definition}[Value Function]
Let $\pi$ be a positive recurrent strategy, and let $\rho$ be a positive constant. For a tree $\tree(Q)$, define
$$\vf_\pi^\rho(\tree(Q)) := \e{\sum_{r = 1}^\tau (X_r(\pi) - \rho) | Q_0 = Q}$$
where $\tau$ is stopping time. Taking the expected value with respect to $\tree(Q)$ gives
$$\vf_\pi^\rho := \e{\vf_\pi^\rho(\tree(Q))}.$$
\end{definition}
We can derive a recursive formula for the value function as follows:
\begin{proposition}\label{prop:vf-recursion} For any positive recurrent strategy $\pi$, positive constant $\rho$, tree $\tree(Q)$,

$$\vf_\pi^{\rho}(\tree(Q)) = \e{(X_1(\pi) - \rho) + \vf_\pi^\rho(Q_1) | Q_0 = Q}$$
\end{proposition}
\begin{theorem}\label{thm:optimization}
Let $\pi$ and $\pi'$ be positive recurrent strategies. Then
\begin{itemize}
    \item $\vf_\pi^\rho = 0$ if and only if $\rho = \rev(\pi)$.
    \item $\rev(\pi') < \rev(\pi)$ if and only if $\vf_\pi^{\rev(\pi')} > \vf_\pi^{\rev(\pi)}$.
\end{itemize}
\end{theorem}
\begin{proof}
From Lemma~\ref{lemma:rev-recurrence} and the assumption $\pi$ is positive recurrent, $\rev(\pi) = \frac{\e{\sum_{r = 1}^\tau X_r(\pi)}}{\e{\tau}}$. Clearly $1 - \rev(\pi) = \frac{\e{\sum_{r = 1}^\tau (1 - X_r(\pi))}}{\e{\tau}}$. Then
\begin{align*}
0 &= \frac{\e{\sum_{r = 1}^\tau X_r(\pi)} \e{\sum_{r = 1}^\tau (1 - X_r(\pi))}}{\e{\tau}} - \frac{\e{\sum_{r = 1}^\tau X_r(\pi)} \e{\sum_{r = 1}^\tau (1 - X_r(\pi))}}{\e{\tau}}\\
&= (1-\rev(\pi)) \e{\sum_{r = 1}^\tau X_r(\pi)} - \rev(\pi) \e{\sum_{r = 1}^\tau (1 - X_r(\pi))}\\
&= \e{\sum_{r = 1}^\tau (X_r(\pi) - \rev(\pi))} \qquad \text{From linearity of expectation,}\\
&= \vf_\pi^{\rev(\pi)}
\end{align*}
The chain of inequalities proofs $\vf_\pi^\rho = 0$ when $\rho = \rev(\pi)$ as desired. For the other direction, observe $\vf_\pi^\rho$ is a strictly decreasing function of $\rho$. Hence there is a unique value for $\rho$ where $\vf_\pi^{\rho}$ vanishes to zero. This proves the first bullet. The second bullet follows from the fact $V_\pi^\rho$ is strictly monotone decreasing in $\rho$.
\end{proof}  
\begin{corollary}\label{cor:opt}
Let $\pi^* \in \arg\max_{\tilde \pi} \rev(\tilde \pi)$. Then $\pi \in \arg\max_{\tilde \pi} \rev(\tilde \pi)$ is optimal if and only if $\pi \in \arg\max_{\tilde \pi} \vf_{\tilde \pi}^{\rev(\pi^*)}$.
\end{corollary}
\begin{proof}
First we prove that if $\pi \in \arg\max_\pi \rev(\pi)$, then $\pi \in \arg\max_\pi \vf_\pi^{\rev(\pi^*)}$. From Theorem~\ref{thm:optimization}, for all strategy $\tilde \pi$,
$$\vf_\pi^{\rev(\pi)} = 0 = \vf_{\tilde \pi}^{\rev(\tilde \pi)} \geq \vf_{\tilde \pi}^{\rev(\pi)}$$
where the first and second equality are the first bullet in the theorem; the inequality is the second bullet and the fact $\rev(\pi) = \rev(\pi^*) \geq \rev(\tilde \pi)$. Since the inequality holds for any $\tilde \pi$, we have $\pi \in \arg\max_{\tilde \pi} \vf_{\tilde \pi}^{\rev(\pi^*)}$. This proves the first part.

For the second part, we proof that if $\pi \in \arg\max_{\tilde \pi} \vf_{\tilde \pi}^{\rev(\pi^*)}$, then $\pi \in \arg\max_{\tilde \pi} \rev(\tilde \pi)$. We already proved that $\vf_{\pi^*}^{\rev(\pi*)} \geq \vf_\pi^{\rev(\pi^*)}$. The assumption implies $\vf_{\pi}^{\rev(\pi^*)} \geq \vf_{\pi^*}^{\rev(\pi^*)}$. Hence $\vf_{\pi}^{\rev(\pi^*)} = \vf_{\pi^*}^{\rev(\pi^*)} = 0$ which proves $\pi$ is optimal (Theorem~\ref{thm:optimization}).
\end{proof}

The following is equivalent to Bellman's principle of optimality.

\begin{lemma}
Let $\pi \in \arg\max_{\tilde \pi} \rev(\tilde \pi)$ and assume $\alpha < \frac{3-\sqrt{5}}{2}\approx 0.38$. Then for all $\tree(Q)$:
$$\pi \in \arg\max_{\tilde \pi} \vf_{\tilde \pi}^{\rev(\pi)}(\tree(Q)).$$
\end{lemma}
\begin{proof}
Let $\pi_r$ refer to the action of strategy $\pi$ at round $r$. From Corollary~\ref{cor:opt}, the fact $\pi$ is optimal implies
\begin{align*}
\vf_\pi^{\rev(\pi)} &= \max_{\tilde \pi} \vf_{\tilde \pi}^{\rev(\pi)}\\
&= \e{\max_{\tilde \pi} \vf_{\tilde \pi}^{\rev(\pi)}(\tree(Q_0))}\\
&= \e{\max_{\tilde \pi} \e{(X_1(\tilde \pi) - \rev(\pi)) + \vf_{\tilde \pi}^{\rev(\pi)}(\tree(Q_1))| Q_0}}\\
&= \e{\max_{\tilde \pi_1} \e{(X_1(\tilde \pi) - \rev(\pi)) + \max_{\tilde \pi : \tilde \pi_1} \vf_{\tilde \pi}^{\rev(\pi)}(\tree(Q_1))| Q_0}} 
\end{align*}
The second equality is Proposition~\ref{prop:vf-recursion}. The third equality observes that the optimal strategy for the sub-game starting with seed $Q_1$ is independent of the action taken at round $1$. Hence $\vf_{\pi}^{\rev(\pi)}(\tree(Q)) = \max_{\tilde \pi} \vf_{\tilde \pi}^{\rev(\pi)}(\tree(Q))$ for all $\tree(Q)$.
\end{proof}

To compute the optimal strategy $\pi^* \in \arg\max_{\tilde \pi} \rev(\tilde \pi)$, we can use a similar binary search algorithm from \citet{sapirshtein2016optimal}. We pick some $\rho \in [0, 1]$ as our guess for $\rev(\pi^*)$ and maximize the Markov Decision Process $\max_{\tilde \pi} \vf_{\tilde \pi}^\rho$. Let $\pi$ be the strategy the solver outputs. Then one of the following cases tell us if $\rho$ is a lower bound or an upper bound on the optimal revenue:

\begin{itemize}
    \item The case where $\vf_{\pi}^\rho \geq 0$ witnesses that $\rev(\pi^*) \geq \rho$. To see, recall $\vf_{\pi}^{\rev(\pi)} = 0$ (Theorem~\ref{thm:optimization}). Because $\vf_{\pi}^\rho$ is a strictly decreasing function of $\rho$, we conclude $\rho \leq \rev(\pi) \leq \rev(\pi^*)$ as desired.
    
    \item Te case where $\vf_{\pi}^\rho < 0$ witnesses that $\rev(\pi^*) < \rho$. To see, it suffices to prove the contra-positive: if $\rev(\pi^*) \geq \rho$, then $\vf_\pi^\rho \geq 0$. Assume $\rev(\pi^*) \geq \rho$. Because $\vf_{\pi}^\rho$ is a strictly decreasing function of $\rho$, we conclude $\vf_{\pi^*}^\rho \geq \vf_{\pi^*}^{\rev(\pi*)} = 0$ (Theorem~\ref{thm:optimization}). By assumption, $\vf_{\pi}^\rho = \max_{\tilde \pi} \vf_{\tilde \pi}^\rho \geq \vf_{\pi^*}^\rho \geq 0$ as desired.
\end{itemize}

\section{Conclusion}

We propose a stylized model to study optimal strategic mining in Cryptographic Self-Selection leader election protocols. We consider rational miners that wish to maximize the fraction of blocks they create. The same adversary has been studied in the context of Proof-of-Work blockchains since the discovery of the selfish mining attacks against Bitcoin~\cite{eyal2014majority}. 

Prior work largely classifies existing protocols into two camps: those where sufficiently small miners \emph{cannot} profitably deviate (longest-chain proof-of-work protocols with block reward and longest-chain proof-of-stake protocols with a randomness beacon), and those where arbitrarily small miners can still profitably deviate (longest-chain proof-of-work protocols with transaction fees, longest-chain proof-of-stake protocols without a randomness beacon). Our work classifies blockchains based on cryptographic self-selection with the latter group: we give a closed-form representation for a strategy that outperforms the honest strategy for any amount of stake. 

The key open question left by our work is to nail down the optimal fraction of rounds that a $\beta$-strong strategic miner with an $\alpha$ fraction of the stake can earn. While our work states that this quantity can in principle be determined by performing binary search over infinitely-sized MDPs, actually significant innovation seems to be required to actually perform this search, or even to approximating it computationally-efficiently.

\bibliography{mybib}

\begin{thebibliography}{20}
\providecommand{\natexlab}[1]{#1}
\providecommand{\url}[1]{\texttt{#1}}
\expandafter\ifx\csname urlstyle\endcsname\relax
  \providecommand{\doi}[1]{doi: #1}\else
  \providecommand{\doi}{doi: \begingroup \urlstyle{rm}\Url}\fi

\bibitem[cbe(2021)]{cbeci}
2021.
\newblock URL \url{https://ccaf.io/cbeci/index}.
\newblock Accessed: 2022-1-10.

\bibitem[Arnosti and Weinberg(2022)]{arnosti2022bitcoin}
N.~Arnosti and S.~M. Weinberg.
\newblock Bitcoin: A natural oligopoly.
\newblock \emph{Management Science}, 2022.

\bibitem[Brown-Cohen et~al.(2019)Brown-Cohen, Narayanan, Psomas, and
  Weinberg]{brown2019formal}
J.~Brown-Cohen, A.~Narayanan, A.~Psomas, and S.~M. Weinberg.
\newblock Formal barriers to longest-chain proof-of-stake protocols.
\newblock In \emph{Proceedings of the 2019 ACM Conference on Economics and
  Computation}, pages 459--473, 2019.

\bibitem[Cachin et~al.(2001)Cachin, Kursawe, Petzold, and
  Shoup]{cachin2001secure}
C.~Cachin, K.~Kursawe, F.~Petzold, and V.~Shoup.
\newblock Secure and efficient asynchronous broadcast protocols.
\newblock In \emph{Annual International Cryptology Conference}, pages 524--541.
  Springer, 2001.

\bibitem[Carlsten et~al.(2016)Carlsten, Kalodner, Weinberg, and
  Narayanan]{carlsten2016instability}
M.~Carlsten, H.~Kalodner, S.~M. Weinberg, and A.~Narayanan.
\newblock On the instability of bitcoin without the block reward.
\newblock In \emph{Proceedings of the 2016 ACM SIGSAC Conference on Computer
  and Communications Security}, pages 154--167, 2016.

\bibitem[Chen and Micali(2019)]{chen2019algorand}
J.~Chen and S.~Micali.
\newblock Algorand: A secure and efficient distributed ledger.
\newblock \emph{Theoretical Computer Science}, 777:\penalty0 155--183, 2019.

\bibitem[Eyal and Sirer(2014)]{eyal2014majority}
I.~Eyal and E.~G. Sirer.
\newblock Majority is not enough: Bitcoin mining is vulnerable.
\newblock In \emph{International conference on financial cryptography and data
  security}, pages 436--454. Springer, 2014.

\bibitem[Ferreira and Weinberg(2021)]{ferreira2021proof}
M.~V. Ferreira and S.~M. Weinberg.
\newblock Proof-of-stake mining games with perfect randomness.
\newblock In \emph{Proceedings of the 22nd ACM Conference on Economics and
  Computation}, pages 433--453, 2021.

\bibitem[Ferreira et~al.(2021)Ferreira, Moroz, Parkes, and
  Stern]{ferreira2021dynamic}
M.~V. Ferreira, D.~J. Moroz, D.~C. Parkes, and M.~Stern.
\newblock Dynamic posted-price mechanisms for the blockchain transaction-fee
  market.
\newblock In \emph{Proceedings of the 3rd ACM Conference on Advances in
  Financial Technologies}, pages 86--99, 2021.

\bibitem[Gilad et~al.(2017)Gilad, Hemo, Micali, Vlachos, and
  Zeldovich]{gilad2017algorand}
Y.~Gilad, R.~Hemo, S.~Micali, G.~Vlachos, and N.~Zeldovich.
\newblock Algorand: Scaling byzantine agreements for cryptocurrencies.
\newblock In \emph{Proceedings of the 26th symposium on operating systems
  principles}, pages 51--68, 2017.

\bibitem[Karantias et~al.(2020)Karantias, Kiayias, and
  Zindros]{karantias2020proof}
K.~Karantias, A.~Kiayias, and D.~Zindros.
\newblock Proof-of-burn.
\newblock In \emph{International conference on financial cryptography and data
  security}, pages 523--540. Springer, 2020.

\bibitem[Kelsey et~al.(2019)Kelsey, Brand{\~a}o, Peralta, and
  Booth]{kelsey2019reference}
J.~Kelsey, L.~T. Brand{\~a}o, R.~Peralta, and H.~Booth.
\newblock A reference for randomness beacons: Format and protocol version 2.
\newblock Technical report, National Institute of Standards and Technology,
  2019.

\bibitem[Kiayias et~al.(2017)Kiayias, Russell, David, and
  Oliynykov]{kiayias2017ouroboros}
A.~Kiayias, A.~Russell, B.~David, and R.~Oliynykov.
\newblock Ouroboros: A provably secure proof-of-stake blockchain protocol.
\newblock In \emph{Annual international cryptology conference}, pages 357--388.
  Springer, 2017.

\bibitem[Micali et~al.(1999)Micali, Rabin, and Vadhan]{micali1999verifiable}
S.~Micali, M.~Rabin, and S.~Vadhan.
\newblock Verifiable random functions.
\newblock In \emph{40th annual symposium on foundations of computer science
  (cat. No. 99CB37039)}, pages 120--130. IEEE, 1999.

\bibitem[Nakamoto(2008)]{nakamoto2008bitcoin}
S.~Nakamoto.
\newblock Bitcoin: A peer-to-peer electronic cash system.
\newblock \emph{Decentralized Business Review}, page 21260, 2008.

\bibitem[Rabin(1983)]{rabin1983transaction}
M.~O. Rabin.
\newblock Transaction protection by beacons.
\newblock \emph{Journal of Computer and System Sciences}, 27\penalty0
  (2):\penalty0 256--267, 1983.

\bibitem[Ren and Devadas(2016)]{ren2016proof}
L.~Ren and S.~Devadas.
\newblock Proof of space from stacked expanders.
\newblock In \emph{Theory of Cryptography Conference}, pages 262--285.
  Springer, 2016.

\bibitem[Roughgarden(2021)]{roughgarden2021transaction}
T.~Roughgarden.
\newblock Transaction fee mechanism design.
\newblock \emph{ACM SIGecom Exchanges}, 19\penalty0 (1):\penalty0 52--55, 2021.

\bibitem[Sapirshtein et~al.(2016)Sapirshtein, Sompolinsky, and
  Zohar]{sapirshtein2016optimal}
A.~Sapirshtein, Y.~Sompolinsky, and A.~Zohar.
\newblock Optimal selfish mining strategies in bitcoin.
\newblock In \emph{International Conference on Financial Cryptography and Data
  Security}, pages 515--532. Springer, 2016.

\bibitem[Watson and Galton(1875)]{watson1875probability}
H.~W. Watson and F.~Galton.
\newblock On the probability of the extinction of families.
\newblock \emph{The Journal of the Anthropological Institute of Great Britain
  and Ireland}, 4:\penalty0 138--144, 1875.

\end{thebibliography}

\newpage
\appendix

\section{Probability Theory Background}\label{app:exp}
\begin{lemma}\label{lemma:min-exp}
Let $X_1, X_2, \ldots, X_n$ be independent random variables where $X_i$ is a copy from $\Exp{\alpha_i}$ where $\alpha_i$ is a positive constant. Then $\min_{i \in [n]} \{X_i\}$ is identically distributed to $\Exp{\sum_{i = 1}^n \alpha_i}$.
\end{lemma}
\begin{proof}
The proof follows from computing the probability $\pr{\min X_i \leq x}$:
\begin{align*}
\pr{\min X_i \leq x} &= 1 - \pr{\land_{i = 1}^n \{X_i > x\}} = 1 - \prod_{i = 1}^n \pr{X_i > x}\\
&= 1 - \prod_{i = 1}^n e^{-\alpha_i x}\\
&= 1 - e^{-x \sum_{i = 1}^n \alpha_i}.
\end{align*}
The last line witness $\min X_i$ is exponentially distributed with weight $\sum_{i = 1}^n \alpha_i$.
\end{proof}
\begin{lemma}\label{lemma:exp-inequality}
Let $X$ and $Y$ be drawn independently from an exponential distributions with rate $\alpha_X$ and $\alpha_Y$ respectively. Then
$$\pr{X < Y} = \frac{\alpha_X}{\alpha_X + \alpha_Y}.$$
\end{lemma}
\begin{proof}
We have as follows:
\begin{align*}
    \pr{X < Y} &= \int_0^\infty f_Y(y)\pr{X < y} dy \\
    &= \int_0^\infty \alpha_Y e^{-\alpha_Y y} (1 - e^{-\alpha_X y}) dy \\
    &= \int_0^\infty \alpha_Y e^{-\alpha_Y y} dy - \int_0^\infty \alpha_Y e^{-(\alpha_Y  + \alpha_X)y} dy \\
    &= 1 - \frac{\alpha_Y}{\alpha_X + \alpha_Y} \\
    &= \frac{\alpha_X}{\alpha_X + \alpha_Y}
\end{align*}
\end{proof}

\begin{corollary}\label{cor:minExps}
Let $X_1,\ldots, X_n$ be drawn independently from exponential distributions with rates $\alpha_1,\ldots, \alpha_n$, respectively. Then $\Pr[X_i = \min_{j\in [n]}\{X_j\}] = \frac{\alpha_i}{\sum_{j=1}^n \alpha_j}$.
\end{corollary}
\begin{proof}
We prove this by induction, using Lemmas~\ref{lemma:min-exp} and~\ref{lemma:exp-inequality}. As a base case, the claim is clearly true when $n=1$, for all $\alpha_1$. Now as an inductive hypothesis, assume that the claim is true for some $n$, and all $\alpha_1,\ldots,\alpha_n$. We now consider the case of $n+1$ and any $\alpha_1,\ldots, \alpha_{n+1}$.

By Lemma~\ref{lemma:min-exp}, $\min_{j\in [n+1]\setminus i} \{X_j\}$ is distributed according to an exponential of rate $\sum_{j\neq i} \alpha_j$. By Lemma~\ref{lemma:exp-inequality}, the probability that $X_{i} = \min_{j \in [n+1]} \{X_j\} = \frac{\alpha_i}{\sum_{j=1}^{n+1} \alpha_j}$, as desired. This argument holds for any $i$, and completes the inductive step.
\end{proof}

\begin{lemma}[Memorylessness Property]\label{lemma:memoryless}
Let $X$ be drawn from an exponential distribution (with any rate $\alpha$), then for any $n, m \geq 0$,
$$\pr{X > n + m | X \geq m} = \pr{X > n}.$$
\end{lemma}
\begin{proof}
We have as follows:
\begin{align*}
    \pr{X > n + m | X \geq m} &= \frac{\pr{X > n + m, X \geq m}}{\pr{X \geq m}} \\
    &= \frac{\pr{X > n + m}}{\pr{X \geq m}} \\
    &= \frac{1 - (1 - e^{-\alpha (n + m)})}{1 - (1 - e^{-\alpha m})} \\
    &= \frac{e^{-\alpha (n + m)}}{e^{-\alpha m}} \\
    &= e^{-\alpha n} \\
    &= \pr{X > n}
\end{align*}
\end{proof}

\begin{lemma}[Strong Law of Large Numbers]\label{lemma:slln}
Let $X$ be a random variable. Let $X_1, X_2, \ldots, X_n$ be independent copies of $X$. Then $\pr{\lim_{n \to \infty} \frac{1}{n} \sum_{i = 1}^n X_i = \e{X}} = 1$.
\end{lemma}

\section{Omitted Proofs from Section~\ref{sec:background}}\label{app:background}

\begin{lemma}\label{lem:scoringg}
Let $g(\cdot)$ be a monotone increasing function with domain $(0,1)$. Define $S^g(x,\alpha):=g(x^{1/\alpha})$.
Then, $S^g(\cdot, \cdot)$ is a canonical balanced scoring function. 
\end{lemma}
\begin{proof}
First, we show that $S^g(\cdot,\cdot)$ is a balanced scoring function. To see this, we observe that:
\begin{align*}
    \Pr_{X_1,\ldots, X_n \leftarrow U([0,1]^n)}\left[\arg\min_{i\in [n]}\left\{S^g(X_i,\alpha_i)\right\} = j\right] &= \Pr_{X_1,\ldots, X_n \leftarrow U([0,1]^n)}\left[\arg\min_{i \in [n]}\{g(X_i^{1/\alpha_i})\} = j\right]\\
    &=\Pr_{X_1,\ldots, X_n \leftarrow U([0,1]^n)}\left[\arg\max_{i \in [n]}\{X_i^{1/\alpha_i}\} = j\right]\\
    &=\Pr_{X_1,\ldots, X_n \leftarrow U([0,1]^n)}\left[\arg\max_{i \in [n]}\{\ln(X_i^{1/\alpha_i})\} = j\right]\\
    &=\Pr_{X_1,\ldots, X_n \leftarrow U([0,1]^n)}\left[\arg\min_{i \in [n]}\{-\ln(X_i^{1/\alpha_i})\} = j\right]\\
    &=\frac{\alpha_j}{\sum_{i=1}^n \alpha_i}\\
\end{align*}

Above, the first line follows by definition of $S^g(\cdot, \cdot)$. The second line follows as $g(\cdot)$ is monotone decreasing. The third follows as $\ln(\cdot)$ is a monotone increasing function. The fourth line follows trivially. The final line follows as $-\ln(X_i^{1/\alpha_i})$ is distributed according to an exponential with rate $\alpha_i$ by Lemma~\ref{lem:exponential}, and Corollary~\ref{cor:minExps} (which states that the minimum of $Y_1,\ldots, Y_n$, when each $Y_i$ is drawn independently from $\Exp{\alpha_i}$ is equal to $X_i$ with probability $\alpha_i$, for all $i$)

To see that $S^g(\cdot, \cdot)$ is canonical, we first observe that because $g(\cdot)$ is monotone decreasing, $S^g(\cdot,\alpha)$ is monotone decreasing for all $\alpha$. To see the second bullet, we simply observe the following facts for any $x$:
\begin{align*}
\Pr_{X \leftarrow U([0,1])}[S^g(X,\sum_{i=1}^n \alpha_i) > g(x)] &= \Pr_{X \leftarrow U([0,1])}[g(X^{1/\sum_{i=1}^n \alpha_i}) > g(x)]\\
&= \Pr_{X \leftarrow U([0,1])}[X^{1/\sum_{i=1}^n \alpha_i} < x]\\
&= \Pr_{X \leftarrow U([0,1])}[X < x^{\sum_{i=1}^n \alpha_i}]\\
&= x^{\sum_{i=1}^n \alpha_i}.
\end{align*}
The first line follows by definition of $S^g(\cdot,\cdot)$. The second follows as $g$ is monotone decreasing. The third follows as both $X, x > 0$, and the final line follows as $X$ is drawn uniformly from $[0,1]$. For similar reasons, we have:

\begin{align*}
\Pr_{\vec{X} \leftarrow U([0,1])^n}[\min_{i=1}^n \{S^g(X_i, \alpha_i)\} > g(x)] &= \Pr_{\vec{X} \leftarrow U([0,1])^n}[\min_{i=1}^n \{g(X_i^{1/\alpha_i})\} > g(x)]\\
&= \Pr_{\vec{X} \leftarrow U([0,1])^n}[\max_{i=1}^n \{X_i^{1/\alpha_i}\} < x]\\
&= \Pr_{\vec{X} \leftarrow U([0,1])^n}[X_i^{1/\alpha_i} <  x,\ \forall i]\\
&=\Pr_{\vec{X} \leftarrow U([0,1])^n}[X_i <  x^\alpha_i,\ \forall i]\\
&= x^{\sum_{i=1}^n \alpha_i}.
\end{align*}

Therefore, we see that for all $n$, and all $\langle \alpha_1,\ldots, \alpha_n\rangle$, the distributions of $S^g(X,\sum_{i=1}^n \alpha_i)$ and $\min_{i=1}^n\{S^g(X_i,\alpha_i)\}$ are identical. 

\end{proof}

For example, the canonical scoring rule we use for our analysis is $S^g(\cdot,\cdot)$ where $g(x):=-\ln(x)$. The canonical scoring rule used in~\cite{chen2019algorand} is $S^h(\cdot,\cdot)$ where $h(x):=1-x$ (where $\alpha = 1$ denotes that the account owns a single coin).

\begin{lemma}\label{lem:scoringuseful}
Let $S$ be a canonical balanced scoring function, and define $g(\cdot):=S(\cdot, 1)$. Then $S=S^g$.
\end{lemma}
\begin{proof}
The proof follows from three simple steps: (a) we show that $S(\cdot, 1/n) = S^g(\cdot,1/n)$ for all integers $n$. Then, we use this to show that $S(\cdot,c/n) = S^g(\cdot, c/n)$, for all integers $c$. This concludes the proof for all rational $\alpha$, which is all we consider.

We now execute the first step. Observe that, because $S$ is canonical, we know exactly what the distribution of $S(X,1/n)$ must be when $X\leftarrow U([0,1])$. Indeed, we must have, for all $x$ (for notational convenience below defining an inverse, we let $S_{1/n}(x):=S(x,1/n)$):
\begin{align*}
\Pr_{\vec{X} \leftarrow U([0,1])^n}[\min_{i=1}^n \{S(X_i,1/n)\} > x] &= \Pr_{X \leftarrow U([0,1])}[S(X,1) > x]\\
\Rightarrow \Pr_{X \leftarrow U([0,1])}[S(X,1/n)> x]^n &= \Pr_{X \leftarrow U([0,1])}[g(X) > x]\\
\Rightarrow \Pr_{X \leftarrow U([0,1])}[S(X,1/n) > x] &= \Pr_{X \leftarrow U([0,1])}[g(X) > x]^{1/n}\\
&= \Pr_{X \leftarrow U([0,1])}[X < g^{-1}(x)]^{1/n}\\
&= (g^{-1}(x))^{1/n}\\
\Rightarrow (g^{-1}(x))^{1/n} &=\Pr_{X \leftarrow U([0,1])}[S(X,1/n)> x]\\
&= \Pr_{X \leftarrow U([0,1])}[X < S_{1/n}^{-1}(x)]\\
&= S^{-1}_{1/n}(x)\\
\Rightarrow S_{1/n}(x) &= g(x^n),\ \forall\ x.
\end{align*}

We now execute the second step, which has nearly identical calculations. 

\begin{align*}
\Pr_{X \leftarrow U([0,1])}[S(X,c/n) > x]&=\Pr_{\vec{X} \leftarrow U([0,1])^c}[\min_{i=1}^c \{S(X_i,1/n)\} > x] \\
&= \Pr_{X \leftarrow U([0,1])}[S(X,1/n) > x]^c\\
&= \Pr_{X \leftarrow U([0,1])}[g(X^{n}) > x]^c\\
&=  \Pr_{X \leftarrow U([0,1])}[X< g^{-1}(x)^{1/n}]^c\\
&= g^{-1}(x)^{c/n}\\
\Rightarrow g^{-1}(x)^{c/n} &= \Pr_{X \leftarrow U([0,1])}[S(X,c) > x]\\
&= \Pr_{X \leftarrow U([0,1])}[X < S_{c/n}^{-1}(x)]\\
&= S_{c/n}^{-1}(x)\\
\Rightarrow S_{c/n}(x) &= g(x^{n/c}),\ \forall\ x.
\end{align*}
\end{proof}

\begin{proof}[Proof of Proposition~\ref{prop:scoring}]
By Lemma~\ref{lem:scoringuseful}, we know that both $S$ and $S'$ are of the form $S^g$ and $S^h$ for some monotone decreasing functions $g, h$. We will use this property to couple outcomes of the two games.

First, we need to define the bijective mapping for each player. The mapping we will use is simple: in each game, split your stake exactly the same way. When choosing which credentials to broadcast, observe that in both games player $i$ has some information available to them (they see the credentials of all accounts they control, plus some other credentials of other players). Then given a strategy for the first game, we can define a strategy for the second game: during round $r$, broadcast the credential of player $i$ if and only if player $i$ broadcasts its credential in the first game.

Next, we need to couple the two games and claim that under this coupling, for all $r$, the leader in both games is the same. We will couple the games so that $\cred_r^j$ is the same for all rounds $r$ and accounts $j$. 

Now, observe that because we have mapped strategies of every player to one which distributes their stake identically among accounts, and because we have coupled the games so that $\cred_r^j$ is the same for all rounds $j$ and accounts $j$, that for all players $i$, the information available to player $i$ is identical in each game. Therefore, player $i$ will choose to broadcast exactly the same set of credentials. The only remaining step is to confirm that the same leader will be selected in each round because both $S$ and $S'$ are canonical.

Indeed, observe that among the set $B_r$ of broadcast credentials, the winner in the CSSPA with $S$ is exactly:
\begin{align*}
    \arg\min_{j \in B_r} \{S(\cred_r^j,\alpha_j)\} &= \arg\min_{j \in B_r}\{ g( ( \cred_r^j)^{1/\alpha_j} )\}\\
    &=\arg\max_{j \in B_r}\{( \cred_r^j)^{1/\alpha_j} )\}
\end{align*}

The first line follows as $S= S^g$, and the second line follows as $g(\cdot)$ is monotone decreasing. By exactly the same reasoning, we have that the winner in CSSPA with $S'$ is:

\begin{align*}
    \arg\min_{j \in B_r} \{S'(\cred_r^j,\alpha_j)\} &= \arg\min_{j \in B_r}\{ h( ( \cred_r^j)^{1/\alpha_j} )\}\\
    &=\arg\max_{j \in B_r}\{( \cred_r^j)^{1/\alpha_j} )\}
\end{align*}

Therefore, we have shown a mapping between strategies, and a coupling between outcomes, such that in each round the leader in both games is the same. This completes the proof.
\end{proof}

\section{Omitted Proofs from Section~\ref{sec:recurrence}}\label{app:recurrence}

\begin{proof}[Proof of Lemma~\ref{lemma:sampling}]
Consider instead the finite stochastic process $X_1, \ldots, X_n$ where each $X_i$ is an i.i.d. copy from $\Exp{\frac{\alpha}{n}}$. For all $i \in [n]$, define the random variables
$$Y_i = \min^{(i)}(\{X_1, \ldots, X_n\}), \qquad Z_i = Z_{i-1} + \Exp{\alpha-\frac{(i-1)\alpha}{n}},$$
and let $Y_0 = Z_0 = 0$. 
\begin{claim}\label{claim:sampling-finite}
For all $i \in [n]$, $Y_i$ is identically distributed to $Z_i$.
\end{claim}
\begin{proof}
The proof is by induction on $i \geq 0$. Assume for $i \geq 0$, $Z_i$ is identically distributed to $Y_i$ and observe the base case ($i = 0$) follows by definition. Then, it suffices to show that for any absolute constant $x$,
$$\pr{Y_{i+1} > x} = \pr{Z_{i+1} > x} = \pr{Z_i + \Exp{\alpha-\frac{i\alpha}{n}} > x}.$$
Fix $Y_i = z$ where $z$ is any absolute constant. For the case $x < z$, the fact $Y_{i+1} \geq Y_i$ implies
$$\pr{Y_{i+1} > x | x < z = Y_i} = 1 = \pr{Y_i + \Exp{\alpha-\frac{i\alpha}{n}} > x | x < z = Y_i}.$$
For the case $x \geq z$, let $A = \{j \in [n] | X_j > Y_i\}$ and observe that with probability 1, $|A| = n-i$. Fix $A = S$ for any set $S \subseteq [n]$. Then
\begin{align*}
    \pr{Y_{i+1} > x | x \geq z = Y_i, A = S} &= \prod_{j \in S} \pr{X_j > x | x \geq z = Y_i, A = S}\\
    &= \prod_{j \in S} \pr{X_j > Y_i + (x - Y_i) | x \geq z = Y_i, A = S}\\
    &= \prod_{j \in S} \pr{X_j > x - Y_i | x \geq z = Y_i}\\
    &= \prod_{j \in S} e^{-(x - z)\frac{\alpha}{n}} = e^{-(x-z)(n-i)\frac{\alpha}{n}
} \quad \text{Since $|A| = n-i$,}\\
    &= \pr{\Exp{(n-i)\frac{\alpha}{n}} > x - Y_i | x \geq z = Y_i}\\
    &= \pr{Y_i + \Exp{\alpha - \frac{i\alpha}{n}} > x | x \geq z = Y_i}\\
\end{align*}
The first line observes that $Y_{i+1} > x$ if and only if for all $j \in S$, $X_j > x$. The third line observes that $j \in S$ if and only if $X_j > Y_i$ and invokes the memoryless property (Lemma~\ref{lemma:memoryless}). By the Law of Total Probability and combining the case where $x < z$ and $x \geq z$, we obtain
\begin{align*}
    \pr{Y_{i+1} > x} &= \pr[Y_i]{\pr[A]{\pr{Y_{i+1} > x | Y_i = z, A = S} | Y_i = z}}\\
    &= \pr[Y_i]{\pr{Y_i + \Exp{\alpha - \frac{i\alpha}{n}} > x | Y_i = z}}\\
    &= \pr{Y_i + \Exp{\alpha - \frac{i\alpha}{n}} > x}
\end{align*}
as desired.
\end{proof}
From Claim~\ref{claim:sampling-finite} and taking the limit as $n \to \infty$ proves the statement.
\end{proof}

\end{document}